\documentclass[11pt,a4paper]{article}
\usepackage{etex}
\usepackage[utf8]{inputenc}
\usepackage[TS1,T1]{fontenc}
\usepackage{xspace}
\usepackage[a4paper,tmargin=3truecm,bmargin=3truecm,rmargin=2.2truecm,lmargin=2.2truecm]{geometry}
\usepackage{amssymb,amsmath,mathtools,cases}
\usepackage[all]{xy}
\usepackage{enumitem}
\usepackage[english]{babel}
\usepackage{hyperref}
\usepackage{lmodern}

\usepackage{amsthm,thmtools}
\newtheorem{theorem}{Theorem}[section]
\newtheorem{proposition}[theorem]{Proposition}
\newtheorem{lemma}[theorem]{Lemma}
\newtheorem{corollary}[theorem]{Corollary}

\theoremstyle{definition}

\newtheorem{definition}[theorem]{Definition}

\theoremstyle{remark}
\newtheorem{remark}[theorem]{Remark}
\hypersetup{
pdfauthor={B. Iochum, T. Masson, A. Sitarz},
pdftitle={kappa-deformation, affine group and spectral triples},
pdfsubject={},
pdfcreator={pdflatex},
pdfproducer={pdflatex},
pdfkeywords={}
}

\hypersetup{
plainpages=false,
colorlinks=true,
linkcolor=black, 
anchorcolor=black, 
citecolor=black, 
urlcolor=black, 
menucolor=black, 
filecolor=black, 
bookmarksopen=true,
bookmarksnumbered=true}

\DeclarePairedDelimiter\abs{\lvert}{\rvert}
\DeclarePairedDelimiterX\scalprod[2]{\langle}{\rangle}{#1,#2}
\DeclarePairedDelimiter\norm{\lVert}{\rVert}

\DeclareMathOperator{\tr}{tr} % trace
\DeclareMathOperator{\Tr}{Tr}	 %% trace

\DeclareMathOperator{\Aut}{Aut}

\newcommand\cdotaction{\mathbin{\cdot}}
\newcommand\dd{\text{\textup{d}}} % up d
\newcommand\regular{\text{reg}}
\newcommand\FourierDR{\mathcal{E}_{\exp}}
\newcommand{\bbbone}{{\text{\usefont{U}{dsss}{m}{n}\char49}}}
\newcommand{\vc}{\vcentcolon =} %%  :=  in definition 
\newcommand{\Trw}{\Tr_\omega}	  %% a Dixmier trace 

\newcommand\gR{{\mathbb R}}
\newcommand\hgR{\widehat{\gR}}
\newcommand\gC{{\mathbb C}}
\newcommand\gN{{\mathbb N}}
\newcommand\gZ{{\mathbb Z}}
\newcommand\gT{{\mathbb T}}
\newcommand\gS{{\mathbb S}}

\newcommand{\lieG}{\mathfrak{g}}
\newcommand{\frakM}{\mathfrak{M}}
\newcommand{\frakN}{\mathfrak{N}}
\newcommand\algA{\mathcal{A}}
\newcommand\caH{\mathcal{H}}
\newcommand\caK{\mathcal{K}}
\newcommand\caD{\mathcal{D}}
\newcommand\caF{\mathcal{F}}
\newcommand\caL{\mathcal{L}}
\newcommand\caS{\mathcal{S}}
\newcommand\caP{\mathcal{P}}

\newcommand{\ualpha}{\underline{\alpha}}
\newcommand{\ubeta}{\underline{\beta}}

\newcommand\fab{{\hat{f}}}
\newcommand\gab{{\hat{g}}}

\newcommand{\ualphaab}{\underline{\hat{\alpha}}}
\newcommand{\ubetaab}{\underline{\hat{\beta}}}
\newcommand{\bbboneab}{{\hat{\bbbone}}}

\newcommand\psiab{\widehat{\psi}}
\newcommand{\astab}{\mathbin{\hat{\ast}}}

\newcommand\fabeta{{\tilde{f}}}
\newcommand\gabeta{{\tilde{g}}}

\newcommand{\bbboneabeta}{{\tilde{\bbbone}}}

\newcommand\psiabeta{\widetilde{\psi}}
\newcommand{\astabeta}{\mathbin{\tilde{\ast}}}

\newcommand{\falpbeta}{{\check{f}}}
\newcommand{\galpbeta}{{\check{g}}}

\newcommand{\ualphaalpbeta}{\underline{\check{\alpha}}}
\newcommand{\ubetaalpbeta}{\underline{\check{\beta}}}
\newcommand{\bbbonealpbeta}{{\check{\bbbone}}}

\newcommand{\astalpbeta}{\mathbin{\check{\ast}}}

\newcommand\hphi{\hat{\phi}}

\newcommand\tildpsi{\widetilde{\psi}}
\newcommand\tildvarphi{\widetilde{\varphi}}
\newcommand\tildphi{\widetilde{\phi}}

\numberwithin{equation}{section}

\begin{document}

\title{$\kappa$-deformation, affine group and spectral triples}
\author{B. Iochum, T. Masson\\
Centre de Physique Th\'eorique\\
Aix-Marseille Univ., CNRS UMR 7332, Univ. Sud Toulon Var\\
Case 907 - Campus de Luminy\\
F-13288 Marseille Cedex 9\\[2ex]
A. Sitarz%
\thanks{supported by the grant from The John Templeton Foundation}\\
Institute of Physics, Jagiellonian University,\\
Reymonta 4, 30-059 Krak\'ow, Poland \\
and \\
Copernicus Center for Interdisciplinary Studies \\
ul. S\l{}awkowska 17, 31-016 Krak\'ow, Poland
}
\date{}
\maketitle

\begin{center}
Dedicated to S.L. Woronowicz on his 70th birthday
\end{center}

\bigskip
\begin{abstract}
A regular spectral triple is proposed for a two-dimensional $\kappa$-deformation. It is based on the naturally associated affine group $G$, a smooth subalgebra of $C^*(G)$, and an operator $\caD$ defined by two derivations on this  subalgebra. While $\caD$ has metric dimension two, the spectral dimension of the triple is one. This bypasses an obstruction described in \cite{IochMassSchu11a} on existence of finitely-summable spectral triples for a compactified $\kappa$-deformation.
\end{abstract}

\newpage

%%%%%%%%%%%%%%%%%%%%%%%%%%%%%%%%%%%%%%%%%%%%%%
\section{Introduction}
%%%%%%%%%%%%%%%%%%%%%%%%%%%%%%%%%%%%%%%%%%%%%%

In 1991, Lukierski, Ruegg, Nowicki and Tolstoi \cite{LukiRuegNowi91a,LukiNowiRueg92a} produced a Hopf algebraic deformation of the universal enveloping algebra of the Poincar\'e Lie algebra,  which falls into the general scheme of deformations of the Lorentz group studied and classified in \cite{PodlWoro90a, WoroZakr94a}. An interesting feature of this deformation was that the deformation parameter, called $\kappa$, is not dimensionless and in physical models could be related to a length or energy scale \cite{Maji88a}. This Hopf algebra found later a natural interpretation as a symmetry of noncommutative space, which was interpreted as the $\kappa$-deformation of Minkowski space \cite{MajiRueg94a,Zakr94a}. This model of a noncommutative space has been used in physics for different purposes, see for instance \cite{AmelGubiMarc09a, AmelMarcPran09a, DaszLukiWoro09a, FreiLivi05a, LukiRuegZakr95a}. On the more mathematical side, the $\kappa$-deformed symmetries were used to study bicovariant noncommutative differential calculi on the $\kappa$-Minkowski space \cite{Sita95a}. More recently, using some quantization maps the star product formulation of the $\kappa$-Minkowski algebra have been presented \cite{DabrPiac09a, DurhSita11a, GayrGracVari07a}.  

Since the $\kappa$-deformed Minkowski space is (as an algebra) an enveloping algebra of a solvable Lie algebra, there is a natural Lie group $G$, which appears behind the $\kappa$-Minkowski \cite{Agos07a, DabrPiac09a, DurhSita11a, IochMassSchu11a}. We shall recall later the construction of $G$, which appears to be the real affine group.

The main question considered here is whether a $\kappa$-deformed space is a noncommutative geometry in the sense of Connes \cite{Conn94a, ConnMarc08a}. So far, apart from some early attempts, \cite{DAnd06a,GayrGracVari07a}, this question was investigated in \cite{IochMassSchu11a} for a compactified version of the $\kappa$-deformation, yielding, through an incursion in number theory and dynamical systems, a kind of no-go result. The non-existence of finitely summable spectral triples for the compactified version of the $\kappa$-deformation, which was related to the group algebra of the Baumslag-Solitar groups was, in fact, a consequence of the no-go theorem of Voiculescu. The negative result was valid, however, only for representations quasi-equivalent to the left regular representation of the algebra, thus leaving a possibility for other constructions  \cite{IochMassSchu11a}.

Although the case of the discrete group (like Baumslag-Solitar group and its group $C^\ast$-algebra) has no direct bearings on the case considered here ($C^\ast$-algebra of a Lie group), we show that for the latter there is a possibility to bypass the potential obstruction and construct a candidate for a spectral triple with a smooth subalgebra of $C^\ast(G)$. The Dirac operator is associated to two derivations obtained from two one-parameter groups of automorphisms of $C^\ast(G)$. But, even if the construction looks like those of the noncommutative torus, there is here a drop of spectral dimension. Such a phenomenon has been already observed in Moyal harmonic deformations \cite{GayrWulk11a}. 

In the $\kappa$-deformation of a $n$-dimensional space, the space-time coordinates satisfy the following solvable Lie-algebraic relations:
\begin{align}
\label{commutationkappa}
[x^0,\,x^j]\vc \tfrac{i}{\kappa} \,x^j, \quad [x^j,\,x^k]=0, \quad j,k=1,\dots,n-1.
\end{align}
Here we assume $\kappa>0$.

Using the Baker--Campell--Hausdorff formula, one gets \cite[eq. (2.6)]{KosiMaslLuki98a}
\begin{equation*}
e^{i c_\mu x^\mu}=e^{i c_0 x^0}\, e^{i c'_j \,x^j} \text{ where } 
c'_j\vc \tfrac{\kappa}{c_0}\,(1-e^{-c_0/\kappa}) c_j.
\end{equation*}
Actually, if $[A,B]=sB$, we have the ``braiding identity''
\begin{equation}
\label{braiding}
 e^A \, e^B=e^{(\exp s)B}\,e^A.
\end{equation}
If we want to realize the $x^\mu$'s as selfadjoint (not necessarily bounded) operators on some Hilbert space, the natural way is to pass to the unitaries:
\begin{equation*}
U_\omega \vc e^{i\omega x^0} \text{ and }V_{\vec{k}}\vc e^{-i\sum_{j=1}^{n-1} k_jx^j}
\end{equation*}
with $\omega,k_j \in \gR$, which generate the $\kappa$-Minkowski group considered in \cite{Agos07a}.

If $W(\vec{k},\omega)\vc V_{\vec{k}}\,U_\omega$, one gets as in \cite[eq. (13)]{Agos07a}
\begin{equation}
\label{grouplaw}
 W(\vec{k},\omega) \, W( \vec{k'},\omega')=W(e^{-\omega/\kappa} \vec{k'}+\vec{k},\omega + \omega'),
\end{equation}
which is nothing else but a presentation of a group law, which, for $n=2$, describes the semidirect product of two abelian groups:
\begin{equation}
\label{groupkappa}
G_\kappa\vc \gR \ltimes_\alpha \gR,
\end{equation}
where $\alpha$ is the following group homomorphism, $\alpha: \gR \to \Aut(\gR)$:
\begin{equation*}
\alpha(\omega)k \vc e^{-\omega/\kappa}k, \text{ for any $\omega, k\in \gR$}.
\end{equation*}
$G_\kappa \simeq \gR^*_+ \ltimes \gR$ is the affine group on the real line. From now on we shall consider only the case $n=2$ (as all difficulties concentrate around this case) and, moreover, we can take $\kappa=1$, as one can freely rescale this parameter change after rescaling $x^0$. So we choose $G\vc G_1$.

The paper is organized as follows. In section~\ref{C*-algebra} we consider the $C^\ast$-algebra $C^\ast(G)$ of the affine group $G$. We describe elements in $C^\ast(G)$ as functions for various choices of variables, the original variables of the group $G$ and their Fourier transforms, and we exhibit a natural trace. In section~\ref{smooth algebra} we choose a dense subalgebra $\algA$ of $C^\ast(G)$ which is compatible with two derivations obtained from one-parameter groups of automorphisms of $C^\ast(G)$. Section~\ref{Representations} is devoted to the irreducible representations. In relation to Plancherel formula, we characterize the represented elements of $\algA$ which are Hilbert-Schmidt or trace-class operators on some Hilbert space $\caH \simeq L^2(\gR)$, showing also that the two derivations implement the operator of position and momentum of one-dimensional quantum mechanics. In section~\ref{Spectral triple} we produce explicitly a spectral triple which is regular for a chosen operator $\caD$ such that $\caD^2$ is essentially the Hamiltonian of a one-dimensional harmonic oscillator. 

The fact that $G$ is a not a liminal group plays an important role in our construction of a spectral triple of dimension 1: there are a lot of trace-class elements in the represented algebra, but there are also many others with non-zero and finite Dixmier traces. Our main result Theorem \ref{spectraldim} shows that these values of Dixmier traces are proportional to a (non-faithful) trace on $C^\ast(G)$.

%%%%%%%%%%%%%%%%%%%%%%%%%%%%%%%%%%%%%%%%%%%%%%
\section{\texorpdfstring{The $C^\ast$-algebra}{The C*-algebra}}
\label{C*-algebra}
%%%%%%%%%%%%%%%%%%%%%%%%%%%%%%%%%%%%%%%%%%%%%%

We consider the crossed product group $G = \gR \ltimes \gR$ with group law
\begin{equation*}
(a,b) \cdotaction (a',b') \vc  (a + a', b + e^{-a} b').
\end{equation*}
The unit element is $(0,0)$ and the inverse is $(a,b)^{-1}=(-a, - e^a b)$. The left Haar measure on $G$ is given by $\dd \mu(a,b) \vc e^a \,\dd a \dd b$, while the right Haar measure is $\dd \mu_R(a,b) \vc \dd a \dd b$. This group is not unimodular, and the modular function is $\Delta(a,b) \vc e^{a}$. 

The group $G$ is the affine $ax+b$ group. It is connected, simply connected and exponential. Since it is solvable, and thus amenable, one has $C^\ast_\text{red}(G) = C^\ast(G)$. 

In the following we will mention this algebra as $C^\ast(G)$. 

Notice that other versions of the affine group $ax+b$ over the real numbers are studied in the literature. They can slightly differ from the present one. For instance, the affine group studied in \cite{EymaTerp79a} is not connected, and it contains $G$ as the connected component to the unit element. The group $G$ or its companions have been widely studied \cite{Choi12a, El-H11a, Eyma64a, Fara12a, Khal74a, Rose76a, Schm10a, Zep75a, Diep96a, Sudo00a} and several uses  appeared in physics \cite{DabrPiac09a, DurhSita11a, GayrGracVari07a, Huyn82a,  Piec03a, Moua03a, Zeit10a}.

By construction, the convolution algebra is defined over the space of $L^1(G, \dd\mu)$-functions with the following product:
\begin{equation*}
(\fab \astab \gab) (a,b) \vc  \int_G \dd \mu(a',b')\, \fab(a',b')\, \gab \big( (a',b')^{-1} \cdotaction (a,b)\big), \text{ for any $\fab,\gab \in L^1(G, \dd\mu)$},
\end{equation*}
which, for the group considered, takes the following explicit expression
\begin{align}
(\fab \astab \gab) (a,b) & = \int_{\gR^2} \dd a' \dd b' e^{a'} \fab(a',b')\, \gab(a - a', e^{a'}(b-b')) \label{eq-prod-ab}\\
&=
\int_{\gR^2} \dd a' \dd b' \fab(a-a', b- e^{-(a-a')} b')\, \gab(a', b'). \nonumber
\end{align}
The involution is defined by $\fab^\ast(a,b) \vc  \Delta(a,b)^{-1} \overline{\fab((a,b)^{-1})}$, so
\begin{equation}
\label{eq-involution-ab}
\fab^\ast(a,b) = e^{-a} \,\overline{\fab(-a, -e^a b)}.
\end{equation}

The completion of the space $L^1(G, \dd\mu)$ with respect to the norm obtained from the left regular representation on $L^2(G,\dd\mu)$ gives us the reduced $C^\ast$-algebra, which coincides with the group $C^\ast$-algebra $C^\ast(G)$. The algebra $C^\ast(G)$ is generated by the dense involutive subalgebra $\caD(G)_\ast \simeq \caD(\gR^2)_\ast$ of compactly supported smooth functions on $G$. 

The usual notation $\caD(M)$ designates the space of compactly supported smooth functions on a smooth manifold $M$, while the subscript $\ast$ in $\caD(G)_\ast$ is used to specify the convolution product on $\caD(G)$ given in \eqref{eq-prod-ab}, in order to distinguish it from the pointwise product.

\medskip
Well known results on structure of $C^\ast$-algebras of semidirect product groups \cite{Will07a} show that $C^\ast(G) \simeq \gR \ltimes C^\ast(\gR)$ where the action of $\gR$ on $C^\ast(\gR)$ is induced by the action of $\gR$ on $\gR$, together with some correction factor, which appears when the Haar measure on $\gR$ (the second one) is not invariant under the action of $\gR$ (the first one). The construction of the $C^\ast$-algebra of a semidirect product group, as given in \cite[Prop.~3.11]{Will07a}, gives the same product as in \eqref{eq-prod-ab}, while the involution is
\begin{equation*}
\fab^\ast(a,b) = e^{a} \,\overline{\fab(-a, -e^a b)},
\end{equation*}
which is different from \eqref{eq-involution-ab}. Of course, both presentations are equivalent, as one can easily see on the level of compactly supported smooth functions: the $C^\ast$-algebra $\gR \ltimes C^\ast(\gR)$ is generated by the involutive subalgebra $\caD(\gR, \caD(\gR)_\ast) \simeq \caD(\gR^2)$ of compactly supported smooth functions in the first variable $a \in \gR$ with values in the space of compactly supported smooth functions on the second variable $b \in \gR$, and the map
\begin{align*}
\fab_\ltimes(a,b) = e^{a} \fab(a,b),
\end{align*}
establishes a natural isomorphism of involutive algebras between $\caD(G)_\ast \ni \fab$ and $\fab_\ltimes \in \caD(\gR, \caD(\gR)_\ast)$, which extends to an isomorphism on the $C^\ast$-algebras. 

In the following, we will denote by $f \in C^\ast(G)$ an element of this $C^\ast$-algebra and we will use some explicit presentations of $f$ as functions of different pairs of variables. The first pair of variables is $(a,b) \in \gR^2$ as before, and the corresponding function is denoted by $(a,b) \mapsto \fab(a,b)$. This convention will also be used for subalgebras 
of $C^\ast(G)$.

\medskip
By Fourier transform, the commutative $C^\ast$-algebra $C^\ast(\gR)$ is isomorphic to the $C^\ast$-algebra $C_0(\hgR)$ of continuous functions on $\hgR \simeq \gR$ (the dual group of $\gR$) vanishing at infinity. The $C^\ast$-algebra $\gR \ltimes C_0(\hgR)$ is generated by functions $\fabeta_\ltimes \in \caD(\gR, \FourierDR(\hgR))$, where $\FourierDR(\hgR)$ designates the algebra of functions on $\hgR$ for pointwise multiplication obtained as the Fourier transform of $\caD(\gR)_\ast$. We will use the variable $\beta \in \hgR$. The algebra $\FourierDR(\gR)$ can be characterized as follows \cite[Thm~7.2.2]{Stri94a}:

\begin{proposition}
\label{prop-fouriersmoothcompactfunctions}
A function $\phi$ is in $\FourierDR(\gR)$ if and only if $x \mapsto \phi(x)$ is an entire analytic function on $\gR$ rapidly decreasing at infinity and such that the analytic function $ z \in \gC \mapsto \phi(z)$ is of exponential type: $\exists a > 0$, $\exists c >0$, such that $\abs*{\phi(z)} \leq c \, e^{a\abs*{\Im(z)}},\,\forall z  \in \gC$.
\end{proposition}

\noindent The rapidly decreasing property of $\phi$ at infinity corresponds to the smoothness of its Fourier transform, while the exponential type property corresponds to the compact support of its Fourier transform. In particular,  $\FourierDR(\gR) \subset \caS(\gR)$.

To $\fab_\ltimes \in \caD(\gR, \caD(\gR)_\ast) \subset \gR \ltimes C^\ast(\gR)$ corresponds a function $\fabeta \in \caD(\gR, \FourierDR(\hgR)) \subset \gR \ltimes C_0(\hgR)$ given by
\begin{equation*}
\fabeta(a,\beta) \vc \int_{\gR} \dd b\, \fab_\ltimes(a,b) \,e^{i b \beta},
\end{equation*}
so, for any $\fab \in \caD\big(\gR, \caD(\gR)_\ast\big) \subset C^\ast(G)$,
\begin{equation}
\label{eq-deffabetafromfab}
\fabeta(a,\beta) = e^a \int_{\gR} \dd b\, \fab(a,b) \,e^{i b \beta}
\end{equation}
with inverse transformation given by
\begin{equation*}
\fab(a,b) = \tfrac{1}{2\pi}\, e^{-a}\int_{\hgR} \dd \beta\, \fabeta(a,\beta) \,e^{-i b \beta}.
\end{equation*}

The induced product of $\fabeta, \gabeta \in \caD(\gR, \FourierDR(\hgR))$ is 
\begin{equation}
\label{eq-productabeta}
(\fabeta \astabeta \gabeta)(a,\beta) = \int_\gR \dd a'\; \fabeta(a',\beta)\, \gabeta(a-a', e^{-a'} \beta),
\end{equation}
with involution
\begin{equation}
\label{eq-inv-productabeta}
\fabeta^{\,\ast}(a,\beta) = \overline{\fabeta(-a, e^{-a} \beta)}.
\end{equation}
The factor $e^{a}$ in \eqref{eq-deffabetafromfab} is convenient to simplify this last relation.

The couple $(a, \beta) \in \gR \times \hgR$ is the second pair of variables used 
present an abstract element $f \in C^\ast(G)$ as a function denoted by $(a, \beta) \mapsto \fabeta(a, \beta)$.

\medskip
The induced action of $\gR$ on $C_0(\hgR)$ defining $\gR \ltimes C_0(\hgR)$ is given by $\rho_a(\phi)(\beta) = \phi(e^{-a} \beta)$ for any $\phi \in C_0(\hgR)$, $a \in \gR$ and $\beta \in \hgR$. Let us introduce two copies $C_0^\nu(\gR)$, labeled by $\nu \in \{-, +\}$, of the algebra of continuous functions on $\gR$ vanishing at infinity. We denote by $u \in \gR$ the variable for the functions in $C^\nu_0(\gR)$ and, for $\nu \in \{-, +\}$, we associate to a function $k_\nu \in C^\nu_0(\gR)$ the following function in $C_0(\hgR)$ of the variable $\beta$:
\begin{equation*}
\phi(\beta) \vc \begin{cases}
k_\nu(u) &\text{for $\beta = \nu e^{-u}$}, \\
0    &\text{otherwise.}
\end{cases}
\end{equation*}
Observe that necessarily $\phi(0) = 0$ for any $\nu$.

For each $\nu \in \{-, +\}$, the above map establishes an algebra morphism, so that $C^\nu_0(\gR) \subset C_0(\hgR)$ is an sub $\ast$-algebra, which, moreover, is preserved by the action $\rho$ of $\gR$. This action, expressed in the variable $u$, takes the explicit form $\rho_a(k_\nu)(u) = k_\nu(u+a)$, which is the regular representation of the abelian group $\gR$ on functions on $\gR$. The two crossed product subalgebras $\gR \ltimes_\regular C^\nu_0(\gR) \subset \gR \ltimes C_0(\hgR)$ are isomorphic to $\caK(L^2(\gR))$, the algebra of compact operators on $L^2(\gR)$ (see \cite{Will07a} for instance) and we denote them by $\caK_\nu \vc \gR \ltimes C^\nu_0(\gR)$. 

The direct sum $\caK_{-} \oplus \caK_{+}$ is an ideal in $\gR \ltimes C_0(\hgR)$ of functions in variables $a$ and $\beta$ which vanish at $\beta = 0$. The quotient of $\gR \ltimes C_0(\hgR)$ by $\caK_{-} \oplus \caK_{+}$ could be, on the other hand, identified with $C^\ast(\gR)$ for the variable $a \in \gR$, and the quotient map is $\fabeta \mapsto \fabeta_{| \beta = 0}$. This is summarized in the short exact sequence (see for instance \cite{Diep96a,Sudo00a})
\begin{equation}
\label{eq-secC*}
\xymatrix@1@C=15pt{{0}\, \ar[r] &\, {\caK_{-} \oplus \caK_{+}}\, \ar[r] & \,{C^\ast(G)}\, \ar[r] & \,{C^\ast(\gR)}\, \ar[r] &\, {0}}.
\end{equation}

%%%%%%%%%%%%%%%%%%%%%%%%%%%%%%%%%%%%%%%%%%%%%%
\section{The smooth algebra}
\label{smooth algebra}
%%%%%%%%%%%%%%%%%%%%%%%%%%%%%%%%%%%%%%%%%%%%%%

Using previous notations, we consider the following dense 
$\ast$-subalgebra of $C^\ast(G)$:
\begin{align*}
\algA &\vc \caD(G)_\ast, 
\\
&\phantom{\vcentcolon}= \caD(\gR, \caD(\gR)_\ast) \text{, functions $\fab$ presented in variables $(a,b)$},
\\
&\phantom{\vcentcolon}= \caD(\gR, \FourierDR(\hgR)) \text{, functions $\fabeta$ presented in variables $(a,\beta)$}.
\end{align*}

An abstract element $f \in \algA$ will be presented as a function $\fab$ or a function $\fabeta$, whenever it is more convenient to use one notation or another, bearing in mind that the transformation (\ref{eq-deffabetafromfab}) allows us to pass easily between both notations. In some computations in section \ref{subsec-metricdim}, we will use the following result, which relies on the definition $\algA \vc \caD(G)_\ast$: 

\begin{proposition}[{\cite[Th\'eor\`eme~3.1]{DixmMall78a}}]
\label{eclatement}
Any $f \in \algA$ can be presented as a finite sum of elements $\sum_{i=1}^N g_i \ast h_i$ for $g_i, h_i \in \algA$.
\end{proposition}

%%%%%%%%%%%%%%%%%%%%%%%%%%%%%%%%%%%%%%%%%%%%%%
\subsection{\texorpdfstring{Relation to the $\kappa$-deformed space}{Relation to the kappa-deformation space}}

In order to relate this algebra $\algA$ to the $\kappa$-deformation space, we introduce a third pair of variables to present $f$ as a function $(\alpha, \beta) \mapsto \falpbeta(\alpha, \beta)$ for $(\alpha, \beta) \in \hgR^2$: starting from $\fabeta \in \caD(\gR, \FourierDR(\hgR))$, one can perform a Fourier transform along the variable $a$ and  define
\begin{equation}
\label{eq-fourierabetaversalpbeta}
\falpbeta(\alpha, \beta) \vc \int_\gR \dd a\, \fabeta(a,\beta) \,e^{i a \alpha}
= \int_{\gR^2} \dd a \dd b\; e^a \fab(a,b) \,e^{i a \alpha} \,e^{i b \beta}.
\end{equation}
The inverse relations are given by
\begin{equation}
\label{eq-changevariablesalphabeta}
\fabeta(a,\beta) = \tfrac{1}{2\pi} \int_{\hgR} \dd  \alpha\, \falpbeta(\alpha, \beta) \,e^{-i a \alpha}
\ \text{ and }\ 
\fab(a,b) = \tfrac{1}{(2\pi)^2}\,e^{-a} \int_{\hgR^2} \dd \alpha \dd \beta \,\falpbeta(\alpha, \beta) \,e^{-i a \alpha} \,
e^{-i b \beta}.
\end{equation}
A straightforward computation shows that the product of $\falpbeta$ and $\galpbeta$ is given by
\begin{equation}
\label{eq-productalphabetavariables}
(\falpbeta \astalpbeta \galpbeta)(\alpha, \beta) = \tfrac{1}{2\pi} \int_{\gR\times\hgR} \dd \omega \dd \alpha'\; \falpbeta(\alpha + \alpha', \beta) \, \galpbeta(\alpha, e^{-\omega} \beta) \, e^{-i \omega \alpha'},
\end{equation}
the involution is
\begin{equation*}
\falpbeta^\ast(\alpha, \beta) = \tfrac{1}{2\pi} \int_{\gR\times\hgR} \dd \omega \dd \alpha'\; 
\overline{\falpbeta(\alpha + \alpha', e^{-\omega}\beta)} \, e^{-i \omega \alpha'}.
\end{equation*}

Using Prop.~\ref{prop-fouriersmoothcompactfunctions}, the algebra $\algA$ is, in this pair of variables, given by
\begin{equation*}
\algA = \FourierDR(\hgR, \FourierDR(\hgR)).
\end{equation*}

\medskip
At $\beta=0$, the product \eqref{eq-productalphabetavariables} is just the pointwise product of functions:
\begin{equation}
\label{eq-betazeroinproduct}
(\falpbeta \astalpbeta \galpbeta)(\alpha, 0) = \falpbeta(\alpha, 0)\, \galpbeta(\alpha, 0).
\end{equation}

Using these new variables, let us define $\ualpha$ and $\ubeta$ as the functions $\ualphaalpbeta(\alpha, \beta) \vc  \alpha$ and $\ubetaalpbeta(\alpha, \beta) \vc  \beta$. Although they are not in the the original domain of the product (\ref{eq-productalphabetavariables}), one can see them as elements of the extended algebra \cite[Definition 3.1]{DurhSita11a}. Their left and right multiplication on elements in $\algA$, when expressed in the $(a,b)$ variables are:
\begin{align}
\label{eq-ualphaab-left-right}
(\ualphaab \astab \fab)(a,b) &= i (\partial_a \fab)(a,b) - i  (\partial_b\, b\fab)(a,b),
&
(\fab \astab \ualphaab)(a,b) &= i (\partial_a \fab)(a,b),
\\
\label{eq-ubetaab-left-right}
(\ubetaab \astab \fab)(a,b) &= i (\partial_b \fab)(a,b),
&
(\fab \astab \ubetaab)(a,b) &= i e^{-a} (\partial_b \fab)(a,b).
\end{align}
This means, that although they are not themselves in $\algA$, $\ualpha$ and $\ubeta$ are elements in the multiplier algebra $M(\algA)$ of $\algA$. According to \eqref{eq-ubetaab-left-right}, $\ubeta$ is in fact in the multiplier algebra of $C_c(\gR, \FourierDR(\hgR))$ (compactly supported continuous functions on the variable $a$). The derivative along the variable $a$ in \eqref{eq-ualphaab-left-right} shows that $\ualpha$ is only in the multiplier algebra of $\caD(\gR, \FourierDR(\hgR))$. This explains in turn our choice for the algebra $\algA$.

Using \eqref{eq-changevariablesalphabeta}, we can formally write $\ualpha$ and $\ubeta$ in the variables $(a,b)$ in terms of the Dirac distribution at $0$ and its derivative as
\begin{align*}
\ualphaab(a,b) &= i\, \delta_0'(a)\,\delta_0(b),\qquad\ubetaab(a,b) = i \, \delta_0(a)\, \delta_0'(b).
\end{align*}
These expressions have to be understood as distributions once inserted in the integral \eqref{eq-prod-ab} which defines the product on $\algA$. 

A computation in the multiplier algebra $M(\algA)$ (see \cite[Example 3.8]{DurhSita11a}) shows that
\begin{equation*}
[\ualpha, \ubeta] = i \ubeta
\end{equation*}
which is the relation defining \eqref{commutationkappa} for the $\kappa$-deformed space when $\kappa=1$. In other words, the variables $(\alpha, \beta)$ can be formally identified with the ``variables'' $(x^0, x^1)$ of the $\kappa$-deformation. In these variables, for any $\falpbeta \in \algA$, the preceding relations takes the form:
\begin{align*}
(\ualphaalpbeta \astalpbeta \falpbeta)(\alpha, \beta) &= \alpha \falpbeta(\alpha, \beta) + i (\beta \partial_\beta \falpbeta)(\alpha, \beta),
&
(\falpbeta \astalpbeta \ualphaalpbeta)(\alpha, \beta) &= \alpha \,\falpbeta(\alpha, \beta),
\\
(\ubetaalpbeta \astalpbeta \falpbeta)(\alpha, \beta) &= \beta\, \falpbeta(\alpha, \beta),
&
(\falpbeta \astalpbeta \ubetaalpbeta)(\alpha, \beta) &= \beta \,\sigma(\falpbeta)(\alpha, \beta),
\end{align*}
where we define
\begin{equation}
\sigma(\falpbeta)(\alpha, \beta) \vc  
\tfrac{1}{2\pi} \int_{\gR\times\hgR} \dd \omega \dd \alpha'\, \falpbeta(\alpha + \alpha', \beta) \, e^{- i \omega \alpha'}\, e^{-\omega},
\label{sigma}
\end{equation}
which takes also the forms
\begin{align*}
\sigma(\fab)(a,b) &= e^{-a} \fab(a,b),
&
\sigma(\fabeta)(a,\beta) &= e^{-a} \fabeta(a,\beta).
\end{align*}
The operator $\sigma$ appears as a twist of the algebra $\algA$ (compare \cite[Proposition 4.1]{DurhSita11a}). It can be extended to the functions $\ualpha$ and $\ubeta$, and one gets:
\begin{align*}
&[\ualphaalpbeta, \falpbeta] = i \beta \partial_\beta \falpbeta\,,
&
&\ubetaalpbeta \astalpbeta \falpbeta = \sigma^{-1}(\falpbeta) \astalpbeta \ubetaalpbeta \,,
\\
&\sigma(\ualphaalpbeta) = \ualphaalpbeta + i \,,
&
&\sigma(\ubetaalpbeta) = \ubetaalpbeta\,.
\end{align*}

While $\algA$ is not a unital algebra, its multiplier algebra $M(\algA)$ is, and its unit $\bbbone$ takes the following form (as a distribution) in the different pairs of variables:
\begin{align*}
\bbboneab(a,b) &= \delta_0(a)',\delta_0(b), 
&
\bbboneabeta(a,\beta) &= \delta_0(a),
&
\bbbonealpbeta(\alpha, \beta) &=  1.
\end{align*}

%%%%%%%%%%%%%%%%%%%%%%%%%%%%%%%%%%%%%%%%%%%%%%
\subsection{A trace}
\label{subsec-trace}

The Fourier transform on $\gR$ induces the natural isomorphism of $C^\ast$-algebras $\caF : C^\ast(\gR) \overset{\simeq\,\,}{\to} C_0(\hgR)$ \cite[Prop.~3.1]{Will07a}. As before, we denote by $\alpha \in \hgR$ the variable for functions in $C_0(\hgR)$. The map $\tau_\gR : C_0(\hgR)_+ \rightarrow [0,\infty]$, $\tau_\gR(\phi) \vc \int_{\hgR} \phi(\alpha) \dd \alpha$ is a lower semicontinuous trace \cite[II.6.7.2(v), II.6.8.3(i)]{Blac06a}. If $\rho : C^\ast(G) \to C^\ast(\gR)$ is the quotient map in \eqref{eq-secC*}, we define 
\begin{equation}
\label{eq-deftau}
\tau : C^\ast(G)_+ \to [0,\infty],
\quad\quad
\tau(f) \vc \tfrac{1}{2\pi}\,  \tau_\gR \circ \caF \circ \rho(f).
\end{equation}
Denote by $\frakM_\tau$ the linear span in $C^\ast(G)$ of $\{ f \in C^\ast(G)_+\, \mid \, \tau(f)<\infty \}$.

\begin{lemma}
\label{lemma-tracetau}
$\tau$ is a lower semicontinuous trace on $C^\ast(G)$ such that $\frakM_\tau$ contains $\algA$.

\noindent For any $f \in \algA$, one has
\begin{equation}
\label{eq-expressiontau}
\tau(f) =   \int_\gR\dd b\, \fab(0,b)= \fabeta(0,0) =\tfrac{1}{2\pi} \int_{\hgR} \dd \alpha \,\falpbeta(\alpha, 0).
\end{equation}
In particular,
\begin{equation}
\label{eq-tau(f*f)}
\tau(f^\ast\ast f) = \int_\gR \dd a\, \abs*{\fabeta(a,0)}^2 = \tfrac{1}{2\pi} \int_{\hgR} \dd \alpha \,\abs*{\falpbeta(\alpha, 0)}^2.
\end{equation}
\end{lemma}

\begin{proof}
$\tau$ is lower semicontinuous because $\tau_\gR$ is a lower semicontinuous trace on $C_0(\hgR)_+$ and the maps $\rho$ and $\caF$ are continuous as morphisms of $C^\ast$-algebras. The trace property of $\tau$ is inherited from the trace property of $\tau_\gR$: for any $f \in C^\ast(G)$,
\begin{align*}
\tau(f^\ast \ast f) 
&= \tfrac{1}{2\pi}\,  \tau_\gR \circ \caF \circ \rho(f^\ast \ast f)
= \tfrac{1}{2\pi}\,  \tau_\gR \left( \overline{(\caF \circ \rho)(f)} (\caF \circ \rho)(f) \right)
= \tfrac{1}{2\pi}\,  \tau_\gR \left( (\caF \circ \rho)(f) \overline{(\caF \circ \rho)(f)} \right)\\
&= \tau(f \ast f^\ast).
\end{align*}

For any $f \in \algA$, a computation in the variables $(\alpha,\beta)$ gives $\tau(f^\ast \ast f) = \tfrac{1}{2\pi} \int_{\hgR} \dd \alpha \,(\falpbeta^\ast \astalpbeta \falpbeta)(\alpha, 0) < \infty$. Thanks to Proposition~\ref{eclatement} and the polarization relation
\begin{equation}
\label{eq-polarization}
4 g\ast h^\ast = (g + h) \ast (g + h)^\ast  - (g - h) \ast (g - h)^\ast + i (g + ih) \ast (g + ih)^\ast - i (g - ih) \ast (g - ih)^\ast, \,\,\,
\end{equation}
one has $\algA \subset \frakM_\tau$, and $\tau$ takes the claimed value of \eqref{eq-expressiontau} in $\falpbeta$ and using \eqref{eq-betazeroinproduct}, one gets \eqref{eq-tau(f*f)}. The others are easily deduced.
\end{proof}

Using representations of $G$, other traces on subspaces of $C^\ast(G)$ will be defined in \ref{subsec-tracefromrepresentations}.

%%%%%%%%%%%%%%%%%%%%%%%%%%%%%%%%%%%%%%%%%%%%%%
\subsection{Derivations}
\label{subsec-derivations}

For any $f \in C^\ast(G)$ and $t \in \gR$, we define 
$\sigma_{t}(f) \in C^\ast(G)$ by its expression in variables $(a,b)$:
\begin{equation*}
\sigma_{t}(\fab)(a,b) \vc \Delta^{it}(a,b) \fab(a,b) = e^{ita} \fab(a,b) .
\end{equation*}
Note that the twist $\sigma$ defined in \eqref{sigma} corresponds to $\sigma_{i}$.

\begin{lemma}
\label{lem-sigmat}
$t \mapsto \sigma_{t}$ is a one-parameter group of automorphisms of $C^\ast(G)$.
\end{lemma}

$\sigma_{t}$ is the natural modular automorphism of $C^\ast(G)$ defined by $\Delta$. 

\begin{proof}
This follows directly from \eqref{eq-prod-ab} and the definition of $\sigma_t$, compare also with the proof of \cite[Proposition 4.1]{DurhSita11a}.
\end{proof}

As in the case of $\sigma$, $\sigma_t$ extends to some of the elements in the multiplier:
\begin{equation*}
\sigma_{t}(\ualphaalpbeta) = \ualphaalpbeta + t, 
\qquad 
\sigma_{t}(\ubetaalpbeta) = \ubetaalpbeta\,,
\end{equation*}
so that one can interpret this one-parameter group of automorphisms of $C^\ast(G)$ as the translation in the time-direction in $\kappa$-deformation. 

\noindent This one-parameter group of automorphisms defines a derivation 
\begin{equation}
\delta_1(f) \vc  \tfrac{d \sigma_{t}(f)}{dt}_{|t=0} \label{delta1}
\end{equation}
on the algebra $\algA$, given in all variables by:
\begin{align*}
&\delta_1(\fab)(a,b) = ia \fab(a,b),
\\
&\delta_1(\fabeta)(a,\beta) = ia \fabeta(a,\beta),
\\
&\delta_1(\falpbeta)(\alpha,\beta) = (\partial_\alpha \falpbeta)(\alpha,\beta).
\end{align*}

We saw that the algebra $\gR \ltimes C_0(\hgR)$ is defined using the action $\rho_a(\phi)(\beta) = \phi(e^{-a} \beta)$ of $a \in \gR$ on any $\phi \in C_0(\hgR)$. It is straightforward to check that the corresponding one-parameter group of automorphisms $\gR \ni u \mapsto \rho_{-u}$ of the algebra $C_0(\hgR)$ can be extended into a one-parameter group of automorphisms of the crossed product $C^\ast(G) \simeq \gR \ltimes C_0(\hgR)$ (because $\gR$ is an abelian group and the action of $u\in\gR$ is the same as the action of $a \in \gR$ defining the crossed product). 
\begin{lemma}
The map $u \in \gR \mapsto \eta_u \in \Aut\big(C^\ast(G)\big)$ given explicitly by
\begin{equation*}
\eta_u(\fabeta)(a,\beta) = \fabeta(a, e^{u}\beta), \qquad \eta_u(\fab)(a,b) = e^{-u} \fab(a, e^{-u} b),
\end{equation*}
is a one-parameter group of automorphisms of $C^\ast(G)$. 
\end{lemma}

\noindent Therefore we have a second derivation 
\begin{equation}
\delta_2(f) \vc  \tfrac{d \eta_{u}(f)}{du}_{|u=0}  \label{delta2}
\end{equation}
on $\algA$, which is:
\begin{align*}
&\delta_2(\fab)(a,b) = - \fab(a,b) - b (\partial_b \fab)(a,b),
\\
&\delta_2(\fabeta)(a,\beta) = \beta (\partial_\beta \fabeta)(a,\beta),
\\
&\delta_2(\falpbeta)(\alpha,\beta) = \beta (\partial_\beta \falpbeta)(\alpha,\beta).
\end{align*}
Observe that $\delta_2(\fabeta)$ vanishes at $\beta = 0$, so that $\delta_2(f) \in \caK_{-} \oplus \caK_{+}$ for any $f \in \algA$.  

\begin{lemma}
The derivations $\delta_1$, $\delta_2$ are real, i.e. $\delta_k (f^\ast) = (\delta_k f)^\ast$ for any $f \in \algA$ and $k=1,2$, and they commute:
\begin{equation*}
[\delta_1, \delta_2] = 0.
\end{equation*}
Moreover, 
\begin{equation*}
\tau\big(\delta_k(f)\big) = 0 \text{ for $k=1,2$}.
\end{equation*}
\end{lemma}

The proof is by direct computations. While the derivation $\delta_1$ is the ordinary derivative along the variable $\alpha$, the derivative along $\beta$ is not. Nevertheless it is a twisted derivation on $\algA$ (see also \cite[Theorem 4.2]{DurhSita11a}):
\begin{equation*}
\partial_\beta (\fabeta \astabeta \gabeta) = (\partial_\beta \fabeta) \astabeta \gabeta + \sigma(\fabeta) \astabeta (\partial_\beta \gabeta).
\end{equation*}

%%%%%%%%%%%%%%%%%%%%%%%%%%%%%%%%%%%%%%%%%%%%%%
\section{Representations}
\label{Representations}
%%%%%%%%%%%%%%%%%%%%%%%%%%%%%%%%%%%%%%%%%%%%%%

%%%%%%%%%%%%%%%%%%%%%%%%%%%%%%%%%%%%%%%%%%%%%%
\subsection{Irreducible representations}
\label{subsec-irrep}

The irreducible unitary representations of $G = \gR \ltimes \gR$ are well known \cite{GelfNaim47a, Khal74a}. For each $\nu \in \{+,- \}$, one has an irreducible infinite dimensional unitary representation $\pi_\nu$ of $G$ on $\caH_{\nu} \vc  L^2(\gR, \dd s)$ given by
\begin{equation*}
(\pi_{\pm}(a,b)\, \phi)(s) \vc  e^{\pm i b\, e^{-s}} \phi(s+a).
\end{equation*}

\noindent These two representations naturally induce representations of $C^\ast(G)$, defined for any $f \in L^1(G, \dd \mu)$ by
\begin{align*}
(\pi_{\pm}(f) \,\phi)(s) 
&\vc  \int_G \dd a \dd b\; e^a \fab(a,b) \big(\pi_{\pm}(a,b) \,\phi \big)(s).
\end{align*}
Thus
\begin{numcases}{(\pi_{\nu}(f) \,\phi)(s) = }
\int_{\gR^2} \dd u \dd b\; e^{u-s} \fab(u-s,b)\, e^{\nu i b e^{-s}} \phi(u) & \text{ in variables $(a,b)$}, 
\label{eq-repvariablesabnu}
\\
\int_\gR \dd u\; \fabeta(u-s, \nu e^{-s})\, \phi(u) & \text{ in variables $(a,\beta)$}, 
\label{eq-representationinvariablesabeta}
\\
\tfrac{1}{2\pi} \int_{\gR\times\hgR} \dd u \dd v\; \falpbeta(v, \nu e^{-s}) \,e^{-i v (u-s)}\, \phi(u) & \text{ in variables $(\alpha,\beta)$}.
\label{eq-representationalphabeta}
\end{numcases}
As we will see in a while, these expressions make sense also for $\nu=0$, however, the corresponding representation is reducible. 

The representation $\pi_\pm$ can be extended to $\ualpha$ and $\ubeta$ as elements of the multiplier algebra, and they are represented as unbounded operators
\begin{align*}
(\pi_\pm(\ualpha)\,\phi)(s) &= -i (\partial_s \phi)(s),
&
(\pi_\pm(\ubeta)\,\phi)(s) &= \pm e^{-s} \phi(s).
\end{align*}

The Schwartz kernel of $\pi_{\pm}(f)$ is  
\begin{numcases}{K_{\pi_{\pm}(f)}(s,u) =}
\int_\gR \dd b\; e^{u-s} \,\fab(u-s,b) \,e^{\pm i b e^{-s}} & \text{ in variables $(a,b)$},
\nonumber%\label{eq-Schwartzkernelab}
\\
\fabeta(u-s, \pm e^{-s}) & \text{ in variables $(a,\beta)$},
\nonumber%\label{eq-Schwartzkernelabeta}
\\
\tfrac{1}{2\pi} \int_{\hgR} \dd v\; \falpbeta(v, \pm e^{-s}) \,e^{-i v (u-s)} & \text{ in variables $(\alpha,\beta)$}.
\label{eq-Schwartzkernelalpbeta}
\end{numcases}

\medskip
Apart from the infinite-dimensional representations, there exists also a family $\{\pi_0^p\}_{p \in \gR}$ of one-dimensional irreducible unitary representations of $G$, defined by
\begin{equation*}
\pi_0^p(a,b) \vc  e^{i a p}.
\end{equation*}
This induces the family of one-dimensional representations of $C^\ast(G)$ given in variables $(\alpha, \beta)$: 
\begin{equation*}
\pi_0^p(f) = \falpbeta(p,0).
\end{equation*}

Using the direct integral of the one-dimensional Hilbert spaces for the standard Lebesgue measure on $\gR$, which could be identified with $L^2(\gR, \dd p)$, and the direct integral of the representation,
\begin{equation}
\label{eq-defpi0}
\pi_0 \vc  \int_\gR^\oplus \dd p\; \pi_0^p\,,
\end{equation}
we obtain a representation of $G$ and $C^\ast(G)$ on $L^2(\gR, \dd p)$:
\begin{align*}
(\pi_0(a,b)\hphi)(p) &= e^{i a p}\, \hphi(p),
&
(\pi_0(f) \hphi)(p) &= \falpbeta(p,0) \hphi(p),
\end{align*}
for any $\hphi \in L^2(\gR, \dd p)$.

This representation can also be described on the Hilbert space $\caH_0 \vc  L^2(\gR, \dd s)$, which we take as the image of $L^2(\gR, \dd p)$ under the standard Fourier transform defined on $L^1(\gR, \dd p) \cap L^2(\gR, \dd p)$ by: $\phi(s) \vc \frac{1}{2 \pi} \int_{\gR} \dd p\; \hphi(p) e^{i p s}$. On $L^2(\gR, \dd s)$ we have

\begin{equation*}
\big(\pi_0(a,b)\phi\big)(s) = \phi(s+a)
\end{equation*}
and for any $f \in L^1(G, \dd \mu)$, 
\begin{equation*}
\big(\pi_{0}(f) \,\phi \big)(s) = \int_{\gR^2} \dd u \dd b\; e^{u-s} \fab(u-s,b)\, \phi(u).
\end{equation*}
This relation is exactly \eqref{eq-repvariablesabnu} for $\nu=0$, so that, from now on, $\pi_0$ will be considered as an element of the family of representations $\{ \pi_\nu \}_{\nu \in \{-,0,+\}}$.

It is easy to see that the image of $C^\ast(G)$ by $\pi_0$ is abelian, which is also a consequence of \eqref{eq-betazeroinproduct} combined with \eqref{eq-representationalphabeta} for $\nu=0$.

\medskip
It is known that the space $\{ \pi_-, \pi_+ \}$ is dense in $\widehat{G}$ for the Fell topology of $\widehat{G}$, while $\{ \pi_0^p \mid p\in I\}$ is closed in $\widehat{G}$ (\cite{Fell62a}, \cite{Khal74a}) if and only if $I$ is closed in $\gR$. As we will see later, these two irreducible representations are sufficient to give all non-trivial contributions to the computations of spectral dimension for our proposed spectral triple. 

The representation $\pi_- \oplus \pi_+$ is faithful. This can be shown directly using the expression of the representation in the variables $(a,\beta)$. This is also a consequence of the decomposition of the left regular representation presented in next subsection along copies of this representation $\pi_- \oplus \pi_+$, see \ref{subsec-leftregrep}.

It is well known that $G$ is not liminal \cite{Khal74a} and in fact:
\begin{theorem}
\label{thm-postliminal}
The real affine group $G$ is postliminal and $C^\ast(G)$ is of type I.
\end{theorem}

\begin{proof}
Since any irreducible representation of $C^\ast(G)$ contains the compact operators (\cite[p.~164]{Khal74a} for $\pi_\pm$, obvious for the $\pi_0^p$'s), $C^\ast(G)$ is GCR. By \cite[Theorem~IV.1.5.7, IV.1.5.8]{Blac06a}, this is equivalent to $C^\ast(G)$ of type I and postliminal.
\end{proof}

%%%%%%%%%%%%%%%%%%%%%%%%%%%%%%%%%%%%%%%%%%%%%%
\subsection{Derivations and representations}
\label{subsec-derivationsrepresentations}

For any $\nu \in \{-,0,+\}$, define on $\caH_{\nu} \vc  L^2(\gR, \dd s)$ 
the following unbounded hermitean operators,
\begin{align}
\label{eq-defpartialk}
(\partial_1 \phi)(s) \vc  s\, \phi(s),
\qquad
(\partial_2 \phi)(s) \vc  -i (\partial_s \phi)(s),
\end{align}
which implement on $\caH_\nu$ the action of derivations \eqref{delta1} and \eqref{delta2}:
\begin{lemma}
\label{lem-commutationderivations}
For any $f \in \algA$ and $\nu \in \{-,0,+\}$, one has
\begin{equation*}
[\partial_k, \pi_{\nu}(f) ] = \pi_{\nu}(i \delta_k f). 
\end{equation*}
\end{lemma}

\begin{proof}
In variables $(a,\beta)$, one has for $k=1$ and $\phi$ in the domain of $\partial_1$:
\begin{align*}
\big(\partial_1 \pi_\nu(f) \phi\big)(s) &= \int_{\gR} \dd u \, s \, \fabeta(u-s, \nu e^{-s}) \, \phi(u),
&
\big(\pi_\nu(f) \partial_1 \phi)(s) &= \int_{\gR} \dd u \, u \, \fabeta(u-s, \nu e^{-s})\, \phi(u),
\end{align*}
so that 
\begin{equation*}
\big([\partial_1, \pi_{\nu}(f) ]\phi \big)(s) = \int_{\gR} \dd u \, (s-u) \, \fabeta(u-s, \nu e^{-s}) \, \phi(u) =  \big(\pi_{\nu}(i \delta_1 f)\phi \big)(s).
\end{equation*}
For $k=2$ and $\phi$ in the domain of $\partial_2$:
\begin{align*}
\big(\partial_2 \pi_\nu(f) \phi \big)(s) &= -i \int_{\gR} \dd u \left[ -(\partial_a \fabeta)(u-s,\nu e^{-s}) - \nu e^{-s}\, (\partial_\beta \fabeta)(u-s, \nu e^{-s})  \right] \phi(u),\\
\big(\pi_\nu(f) \partial_2 \phi \big)(s) &= - i \int_{\gR} \dd u\, \fabeta(u-s,\nu e^{-s}) \,(\partial_u \phi)(u)
=  i \int_{\gR} \dd u\, (\partial_a \fabeta)(u-s,\nu e^{-s}) \,\phi(u)
\end{align*}
which leads to
\begin{equation*}
\big([\partial_2,\, \pi_{\nu}(f) ]\,\phi \big)(s) = i \int_{\gR} \dd u \,  (\beta\partial_\beta \fabeta)(u-s, \nu e^{-s}) \phi(u) = \big(\pi_{\nu}(i \delta_2 f)\,\phi\big)(s).
\end{equation*}
\end{proof}

Therefore the derivations $\delta_1, \delta_2$ of $\algA$ can be represented as commutators on each of the three representations $\pi_-$, $\pi_+$, and $\pi_0$. Having in mind, for example, the construction of spectral triples for the noncommutative torus, we shall make this a starting point on a search of Dirac operator.

To end this section, let us observe that although the derivations $\delta_1,\, \delta_2$ commute on $\algA$, the operators implementing them do not, as on $L^2(\gR, \dd s)$ one get $[\partial_1, \partial_2] = -i$. Furthermore, one can directly check that 
\begin{equation*}
\pi_0(\delta_2 f) = 0 \text{  for any }f \in \algA.
\end{equation*}

%%%%%%%%%%%%%%%%%%%%%%%%%%%%%%%%%%%%%%%%%%%%%%
\subsection{The left regular representation}
\label{subsec-leftregrep}

The algebra $C^\ast(G)$ is completely determined by the left regular representation of $G$ on $L^2(G, \dd\mu)$. For  $f \in L^1(G, \dd \mu)$, this representation is given on $\psiab \in L^2(G, e^a \,\dd a \dd b)$ by
\begin{equation*}
(\pi_\regular(f) \psiab)(a,b) = \int_{\gR^2} \dd a' \dd b' e^{a'} \fab(a',b')\, \psiab(a - a', e^{a'}(b-b')).
\end{equation*}
Choosing the variables $(a,\beta)$ on $G$, this could be rewritten as
\begin{equation*}
(\pi_\regular(f) \psiabeta)(a,\beta) = \int_\gR \dd a'\; \fabeta(a',\beta)\, \psiabeta(a-a', e^{-a'} \beta)
\end{equation*}
for any $\fabeta \in \caD(\gR, \FourierDR(\hgR))$ and $\psiabeta \in L^2(G, \tfrac{1}{2\pi}e^{-a} \,\dd a \dd \beta)$, where $\psiabeta$ is defined from $\psiab$ as in \eqref{eq-deffabetafromfab}.

It is known (see for instance \cite{Khal74a}) that the left-regular representation decomposes into irreducible representations involving only $\pi_+$ and $\pi_-$. The explicit decomposition is done as follows.

First, this representation decomposes into two pieces. Take $\beta \in \hgR$ and introduce a pair $(\nu, s)$ with $\nu \in \{-,+\}$ and $s \in \gR$ such that $\beta = \nu e^{-s}$. Let us define two Hilbert spaces
\begin{equation*}
\caH_\pm = L^2(\gR^2, e^{-(a+s)}\dd a \dd s),
\end{equation*}
and two maps
\begin{equation*}
\zeta_\pm: L^2(G, \tfrac{1}{2\pi}e^{-a} \,\dd a \dd \beta) \ni 
   \psiabeta \mapsto \psiabeta^\pm \in \caH_\nu,
\end{equation*}
by the simple change of variables
\begin{equation*}
\psiabeta^\pm(a,s) \vc  \psiabeta(a, \pm e^{-s}).
\end{equation*}
   
\begin{lemma}
The operator $\zeta \vc \zeta_- \oplus \zeta_+ : L^2(G, \tfrac{1}{2\pi}e^{-a} \,\dd a \dd \beta) 
\to \caH_- \oplus \caH_+$ is unitary.
\end{lemma}

\begin{proof}
The surjectivity of $\zeta$ is straightforward, the fact that the map preserves the inner product follows directly from computations:
\begin{equation}
\label{eq-scalarproductLtwoG}
\scalprod{\tildpsi_1}{\tildpsi_2} = 
\tfrac{1}{2\pi} \sum_{\nu = \pm} \int_{\gR^2} \dd a \dd s\; e^{-(a+s)}\, \overline{\tildpsi_1^\nu(a, s)}  \,\,\tildpsi_2^\nu(a, s).
\end{equation}
\end{proof}

Therefore, instead of considering $\pi_\regular$, we may study the representation $\pi_\regular^\zeta \vc \zeta \, \pi_\regular \, \zeta^*$ on $\caH_- \oplus \caH_+$ which actually restricts to $\caH_-$ and $\caH_+$, and \begin{equation}
\label{eq-regularrepresentation}
\left( \pi_{\regular}^\zeta(f)\,\psiabeta^\pm \right)(a,s) 
= \int_\gR \dd a' \fabeta(a', \nu e^{-s})\, \psiabeta^\pm(a-a', s+a').
\end{equation}

\begin{lemma}
For $\varphi \in L^\infty(\gR, \dd x)$, define $\tildvarphi(a,s) \vc  \varphi(a+s)$. Then, for any $\psiabeta \in \caH_+ \oplus \caH_-$ and $\fabeta \in \caD(\gR, \FourierDR(\hgR))$, 
\begin{equation}
\label{eq-functionofucommuteswithrepresentation}
\pi_{\regular}^\zeta(f)\,(\tildvarphi \,\psiabeta^\pm) = \tildvarphi\, \pi_{\regular}^\zeta(f)\,\psiabeta^\pm.
\end{equation}
\end{lemma}

\begin{proof}
This follows directly from (\ref{eq-regularrepresentation}).
\end{proof}

This results means that the (commutative) algebra of functions $\varphi \in L^\infty(\gR, \dd x)$ is contained in the commutant of the representation $\pi_{\regular}^\zeta$. Thus, $\pi_{\regular}^\zeta$ can be decomposed as a direct Hilbert integral along $\gR$. Introducing a further change of variables
\begin{equation*}
\tildphi^\pm(v,s) \vc  e^{-v/2}\,\psiabeta^\pm(v-s,s),
\end{equation*}
we see that $\tildphi^\nu \in L^2(\gR^2, \dd v \dd s)$, and the representation mapped to that Hilbert space becomes
\begin{align*}
\left( \pi_{\regular}^\zeta(f)\,\tildphi^\pm \right)(v,s) 
&\vc  e^{-v/2}\left( \pi_{\regular}(f)\,\psiabeta^\pm \right)(v-s,s)
= e^{-v/2}\int_\gR \dd a' \fabeta(a', \pm e^{-s})\, \psiabeta^\pm(v-(s+a'), s+a') \\
&= e^{-v/2}\int_\gR \dd a' \fabeta(a', \pm e^{-s})\, e^{v/2}\tildphi^\pm(v, s+a')
= \int_\gR \dd u\, \fabeta(u-s, \pm e^{-s})\, \tildphi^\pm(v, u).
\end{align*}
Comparing this last expression with \eqref{eq-representationinvariablesabeta}, we get the direct Hilbert integral decomposition of $\pi_{\regular}^\zeta$ along $v \in \gR$:
\begin{equation}
\pi_{\regular}^\zeta \simeq \int_\gR^{\oplus} \dd v \, (\pi_+ \oplus \pi_-) = \pi_{\regular}^+ \oplus \pi_{\regular}^-\,,
\end{equation}
where we define
\begin{align}
\label{eq-defpiplusminusregular}
\pi_{\regular}^\pm \vc \int_\gR^{\oplus} \dd v \, \pi_\pm\,.
\end{align}
In this decomposition of the left regular representation into irreducible representations, the one-dimensional representations, $\pi^p_0$, do not appear, so only the set $\{ \pi_-, \pi_+ \}$ is the principal series of $G$.

%%%%%%%%%%%%%%%%%%%%%%%%%%%%%%%%%%%%%%%%%%%%%%
\subsection{\texorpdfstring{The representation associated to $\tau$}{The representation associated to tau}}

We associate to $\tau$ a representation which generalizes the GNS construction \cite[II.6.7.3]{Blac06a}:
\begin{proposition}
The representation $\pi_\tau$ associated to the trace $\tau$ defined by \eqref{eq-deftau} is unitarily equivalent to $\pi_0$.
\end{proposition}

\begin{proof}
As in \cite[II.6.7.3]{Blac06a}, let us define
\begin{equation*}
\frakN_\tau \vc \{ f \in C^\ast(G) \, \mid \, \tau(f^\ast \ast f) < \infty \}
\quad\text{and}\quad
N_\tau \vc \{ f \in C^\ast(G) \, \mid \, \tau(f^\ast \ast f) = 0 \}.
\end{equation*}
One has $f \in N_\tau$ if and only if $\tfrac{1}{2\pi}\tau_\gR\left( \overline{(\caF \circ \rho)(f)} (\caF \circ \rho)(f)\right) = \tfrac{1}{2\pi}\int_{\hgR} \dd \alpha \, \abs*{(\caF \circ \rho)(f)}^2(\alpha) =  0$, which is equivalent to $(\caF \circ \rho)(f) = 0$, so that, $\caF$ being an isomorphism, $N_\tau =  \ker\rho = \caK_- \oplus \caK_+$ by \eqref{eq-secC*}. In the same way, $f \in \frakN_\tau$ if and only if $\tfrac{1}{2\pi}\int_{\hgR} \dd \alpha \, \abs*{(\caF \circ \rho)(f)}^2(\alpha) < \infty$, which is equivalent to the fact that $(\caF \circ \rho)(f) \in L^2(\hgR, \tfrac{1}{2\pi}\dd\alpha) \cap C_0(\hgR)$, so that, because $N_\tau$ is the kernel of \eqref{eq-secC*}, the quotient $\frakN_\tau/N_\tau$ can be identified with the subspace of $C_0(\hgR) \simeq C^\ast(\gR)$ of square integrable functions: $\frakN_\tau/N_\tau \simeq L^2(\hgR, \tfrac{1}{2\pi}\dd\alpha) \cap C_0(\hgR)$. Since \eqref{eq-tau(f*f)} shows that the scalar product induced by $\tau$ on $\frakN_\tau/N_\tau$ is the scalar product on $L^2(\hgR, \tfrac{1}{2\pi}\dd\alpha)$, the representation space of $\pi_\tau$ is $\caH_\tau \vc L^2(\hgR, \tfrac{1}{2\pi}\dd\alpha)$.

Performing a Fourier transform $L^2(\hgR, \tfrac{1}{2\pi}\dd\alpha) \rightarrow L^2(\gR, \dd a)$ (which is the inverse Fourier transform defining $\caF : C^\ast(\gR) \rightarrow C_0(\hgR)$), we now characterize $\pi_\tau$ on $\caH_\tau$: in the variable $a \in \gR$, the representation takes the explicit form
\begin{equation}
\label{eq-pitau}
\big(\pi_\tau(f)\psi\big)(a) = \int_{\gR} \dd a' \,\fabeta(a-a', 0)\, \psi(a')
\end{equation}
for any $\psi \in L^2(\gR, \dd a)$. It is shown in section \ref{subsec-irrep} that $\pi_0$ defined in \eqref{eq-defpi0} can be presented as a representation on $\caH_0 \vc  L^2(\gR, \dd s)$ by a Fourier transform $\phi(s) \vc \frac{1}{2 \pi} \int_{\gR} \dd p\; \hphi(p) \,e^{i p s}$. A direct computation gives \eqref{eq-pitau} with the Fourier transform $\psi(a) \vc \frac{1}{2 \pi} \int_{\gR} \dd p\; \hphi(p) \,e^{- i p a}$.
\end{proof}

%%%%%%%%%%%%%%%%%%%%%%%%%%%%%%%%%%%%%%%%%%%%%%
\subsection{Traces from representations}
\label{subsec-tracefromrepresentations}

The representations $\pi_\pm$ are traceable (in the sense of \cite[17.1.6]{Dixm69a}) and the respective traces $\tr_\pm(f)$ of these representations, called also normalized characters of $C^\ast(G)$, are 
\begin{equation}
\tr_\pm(f) \vc  \Tr(\pi_\pm(f)), \quad f \in C^\ast(G)_+.  \label{caractere}
\end{equation}
If they are finite, they are computed in the variables $(a, \beta)$ by the integrals
\begin{align}
\label{eq-expressiontracesplusmoins}
\tr_-(f) = \int_{-\infty}^0 \dd \beta\; \tfrac{1}{\beta} \fabeta(0,\beta),
\qquad
\tr_+(f) = \int_0^{+\infty} \dd \beta\; \tfrac{1}{\beta} \fabeta(0,\beta),
\end{align}
or, in the variable $u$, $\beta = \pm e^{-u}$, by
\begin{equation}
\label{eq-tracefabeta}
\tr_\pm(f) = \int_\gR \dd u\; \fabeta(0, \pm e^{-u}).
\end{equation}

\noindent These traces are finite when $\pi_\pm(f)$ is trace-class (see Proposition~\ref{trace-class}) and computable also using the formula $\tr_\pm(f)= \int_\gR \dd u\; K_{\pi_{\pm}(f)}(u,u)$. 

For $\nu = 0$, we define $\tr_0(f) \vc \Tr (\pi_0(f))$ when $\pi_0(f)$ is trace-class. When $f \in \algA$, the Schwartz kernel of $\pi_{0}(f)$ is $K_{\pi_{\nu}(f)}(s,u) = \fabeta(u-s, 0)$, so is continuous on the diagonal, so that, using \cite[Corollary~3.2]{Bris91a}, its trace should be $\Tr (\pi_0(f)) = \int \dd u\, K_{\pi_{\nu}(f)}(u,u) = \int \dd u\, \fabeta(0, 0)$. But it is finite only for $f \in \algA$ such that $\fabeta(0, 0) = \tau(f) = 0$.

Neither of the above traces $\tr_\nu$, for $\nu \in \{-,0,+\}$, is related to the trace $\tau$ on $C^\ast(G)$ defined in \ref{subsec-trace}. However, $\tau$ is related to the individual traces of the family of the one-dimensional representations $\pi_0^p$ of $G$: if
\begin{equation*}
\tr_0^p(f) \vc \Tr(\pi^p_0(f)) = \int_{\gR^2} \dd a \dd b e^{a} \,\fab(a,b) \,e^{iap},
\end{equation*}
we get
\begin{equation*}
\tau(f) = \int_\gR \dd p\, \tr_0^p(f).
\end{equation*}
In other words, the trace $\tau$ on the algebra is the integration along $\gR$ of the field of traces $p \mapsto \tr_0^p$ defined by the one-dimensional irreducible representations $\pi_0^p$.

%%%%%%%%%%%%%%%%%%%%%%%%%%%%%%%%%%%%%%%%%%%%%%
\subsection{Hilbert-Schmidt and trace-class operators}

A complete characterization of the Hilbert-Schmidt and trace-class operators on the representation spaces $\caH_\pm$ of $\pi_\pm$, is given in \cite{Khal74a}. Here we expose the main results in our notations. 

We denote by $\caL^1(\caH_\pm)$ (resp. $\caL^2(\caH_\pm)$) the space of trace-class operators (resp. Hilbert-Schmidt operators) on $\caH_\pm$ and by $\caL^1(\pm)$ (resp. $\caL^2(\pm)$) the space of couples $S = (S_-, S_+)$ of operators $S_\pm \in \caL^1(\caH_\pm)$ (resp. $S_\pm \in \caL^2(\caH_\pm)$), which are Banach spaces for the norms
\begin{equation*}
\norm*{S}_p^p \vc  \norm*{S_-}_p^p + \norm*{S_+}_p^p,\quad p=1,2.
\end{equation*}

Let us define the unbounded Duflo-Moore operator $\theta$ on $\caH_\pm$ \cite{DuflMoor76a} by
\begin{equation}
(\theta \,\phi)(s) \vc  e^{-s/2} \phi(s). \label{theta}
\end{equation}

\noindent This operator $\theta$ is related to the one-parameter group of automorphisms $\sigma_{t}$ (Lemma~\ref{lem-sigmat}):
\begin{lemma}
\label{lem-commutepinudelta}
When $\phi$ is in the domain of $\theta$, then $\pi_\pm(f)\, \phi$ is also in the domain of $\theta$ for $f \in \algA$ and
\begin{equation*}
\theta \,\pi_\pm(f) = \pi_\pm(\sigma_{-i/2}(f))\, \theta.
\end{equation*}
\end{lemma}
Recall that $\sigma_t(\fabeta)(a,\beta) = e^{ita} \fabeta(a,\beta)$, so that $\theta$ is a realization of the modular factor $\Delta^{1/2}$ on the representation spaces $\caH_\pm$. Iterating this relation, one gets $\pi_\pm(f)\, \theta^2 = \theta^2 \,\pi_\pm(\sigma(f))$.

\begin{proof}
For any $\phi$ in the domain of $\theta$, one has in variables $(a,\beta)$, 
\begin{align*}
\big(\pi_\pm(\sigma_{-i/2}(f))\, \theta \,\phi\big)(s) &= \int_\gR \dd u\; e^{(u-s)/2} \fabeta(u-s, \pm e^{-s})\, e^{-u/2} \,\phi(u) 
 = e^{-s/2} \, \int_\gR \dd u\; \fabeta(u-s, \pm e^{-s})\,  \phi(u) \\
& = \big(\theta \,\pi_\pm(f)\,\phi\big)(s).
\end{align*}
In particular $\pi_\pm(f) \,\phi$ is in the domain of $\theta$.
\end{proof}

Let $f \in \algA$. Since $\Delta^{-1/2} f \in \algA$, define the operators $\caP_\pm(f) \vc  \theta \,\pi_\pm (\Delta^{-1/2} f)$ on $\caH_\pm$ and the Plancherel transformation $\caP$
\begin{equation*}
f \mapsto \caP(f) \vc  (\caP_-(f), \caP_+(f)) 
\end{equation*}
mapping $f$ to a pair of operators on $\caH_- \oplus \caH_+$. 

\begin{proposition}[\cite{Khal74a}]
For any $f\in \algA$, one has $\caP(f) \in \caL^2(\pm)$ and
\begin{equation}
\label{eq-plancherelformula}
\norm*{f}_{L^2(G,\dd \mu)} = \norm*{\caP(f)}_2.
\end{equation}
The application $f \mapsto \caP(f)$ extends to an isometric isomorphism from $L^2(G,\dd \mu)$ onto $\caL^2(\pm)$.
\end{proposition}

The relation \eqref{eq-plancherelformula} is the Plancherel formula for the group $G$. This relation does not use the representations $\pi_0^p$, because, as mentioned before, the $\pi_0^p$'s are weakly contained in the $\pi_\nu$'s. At first glance, the operators $\pi_\pm(f)$ are expected to be the operators used on the right hand side of the Plancherel formula. But the non-unimodularity of $G$ implies that these operators must be replaced by their ``twisted'' versions $\caP_\pm(f)$.

\begin{corollary}
The operator $\caP_\pm : L^2(G,\dd \mu) \rightarrow \caL^2(\caH_\pm)$ is surjective.
\end{corollary}
\noindent This corollary tells us that we know all Hilbert-Schmidt operators on $\caH_\nu$. They are of the form $\caP_\nu(f)$ for some $f \in L^2(G,\dd \mu)$:

\begin{proposition}
Let $f \in \algA$. Then, $\pi_\pm(f)$ is a Hilbert-Schmidt operator if and only if there exists $g \in \algA$ such that
\begin{equation*}
\fabeta(a,\beta) = \sqrt{\abs*{\beta}} \,e^{-a/2} \,\gabeta(a,\beta) \, \text{ for any }(a,\beta) \in \gR \times \hgR.
\end{equation*}
\end{proposition}

\begin{proof}
Both sides of the relations are in $\algA$, and a direct computation shows that $\pi_\pm(f) = \caP_\pm(g)$. Notice that the factor $e^{-a/2}$ is unnecessary to characterize functions $f \in \algA$ such that $\pi_\pm(f)$ is Hilbert-Schmidt since $\gabeta$ is compactly supported in the variable $a$. It is only used to relate $\pi_\pm(f)$ to $\caP_\pm(g)$.
\end{proof}

If $f \in \algA$ is such that $\pi_\pm(f)$ is Hilbert-Schmidt, then $\fabeta(a,0) = 0$, so that $f \in \caK_{-} \oplus \caK_{+}$ (and then of course $\pi_\pm(f)$ is compact).

Several extensions of above results can be made for $p \neq 1,2$ via Hausdorff-Young theorem \cite{ EymaTerp79a,Russ77a}.

\medskip
Let $B(G)$ be the algebra of linear combinations of continuous functions of positive type on $G$ \cite{Eyma64a}. It is generated by the functions of the form $(a,b) \mapsto F_\pm(a,b) \vc  \scalprod{\pi_\pm(a,b) \xi}{\eta}_{\caH_\pm}$ for any $\xi,\eta \in \caH_\pm$. This commutative algebra is a Banach algebra for the norm 
\begin{equation*}
\norm*{F} \vc  \sup_{\substack{f \in L^1(G,\dd \mu) \\ \norm*{f}\leq 1}} \int_G \dd \mu(a,b) \fab(a,b)\, F(a,b)
\end{equation*}
where $\norm*{f}$ is the $C^\ast$-norm on $C^\ast(G)$. Consider the Fourier algebra $A(G) \subset B(G)$ of $G$ generated by the linear combinations of continuous compactly supported functions of positive type on $G$, equipped with the same norm. For the affine group, this algebra is given by $A(G) = B(G) \cap C_0(G)$, where $C_0(G)$ is the algebra of continuous functions on $G$ vanishing at infinity \cite{Khal74a}.

The following theorem describe the elements in $A(G)$ and gives a complete description of trace-class operators on $\caH_- \oplus \caH_+$: 

\begin{theorem}[\cite{Khal74a}]
Any element $F \in A(G)$ can be written as $F = f \ast (\Delta g^\ast)$ where $f,g \in L^2(G, \dd\mu)$, and moreover, $\norm*{F} = \norm*{f}_2\, \norm*{g}_2$.

\noindent Let $S = (S_-, S_+)\in \caL^1(\pm)$. Then the function
\begin{equation*}
F(a,b) \vc  \Tr \big(\pi_-(a,b) S_-\big) + \Tr \big(\pi_+(a,b) S_+\big)
\end{equation*}
belongs to $A(G)$ and satisfies $\norm*{F} = \norm*{S}_1$. The association $S \mapsto F$ is an isometric isomorphism from $\caL^1(\pm)$ onto $A(G)$.
\end{theorem}

More results in \cite{Khal74a} show that the restriction $S_\pm \mapsto \Tr\big(\pi_\pm(a,b) S_\pm\big)$, for $S_\pm \in \caL^1(\caH_\pm)$, characterizes the trace-class operators on $\caH_\pm$ as functions in a subalgebra $A_\pm(G) \subset A(G)$ for which $A(G) = A_-(G) \oplus A_+(G)$; and any $S_\pm \in \caL^1(\caH_\pm)$ can be written as $S_\pm = \caP_\pm(f) \,\caP_\pm(g)$ for $f,g \in L^2(G, \dd\mu)$.

\begin{proposition}
\label{trace-class}
For $f \in \algA$, $(\pi_-\oplus \pi_+)(f)\in \caL^1(\caH)$ if and only if there exist $h_1,h_2\in L^2(G,\dd\mu)$ such that
\begin{equation*}
\fabeta(a, \beta) = \tfrac{\abs*{\beta}}{2\pi} \int_{\gR^3} \dd a' \dd b \dd b'\, e^{a'}\, h_1(a',b')\, \overline{h_2(a+a', e^{-a}b' - b)} \,e^{-i b \beta}.
\end{equation*}
So that if $(\pi_-\oplus \pi_+)(f)\in \caL^1(\caH)$, then $\tau(f) = 0$.

\noindent For $\nu=0$ and $f \in \algA$, $\pi_0(f^\ast \ast f)$ is trace-class if and only if $\tau(f^\ast \ast f) = 0$ and then $\tr_0(f^\ast \ast f) = 0$.
\end{proposition}

\begin{proof}
If $(\pi_-\oplus \pi_+)(f)$ is trace-class for $f\in \algA$, then it defines $F \in A(G)$ by
\begin{equation}
F(a,b) \vc  \sum_\nu\, \Tr( \pi_\nu(a,b) \,\pi_\nu(f)) = \sum_\nu \int_{\gR} \dd u\, \fabeta(-a, \nu e^{-(u+a)}) \,e^{i \nu b e^{-u}},
\label{eq-Fnu}
\end{equation}
and the trace is then given by $F(0,0) = \sum_\nu \int_{\gR} \dd u\, \fabeta (0, \nu e^{-u})$, which is \eqref{eq-tracefabeta}. The relation \eqref{eq-Fnu} can be inverted as
\begin{equation}
\label{eq-fabetafromF}
\fabeta(a,\beta) = \tfrac{\abs*{\beta}}{2\pi} \int_{\gR} \dd b\, F(-a, -e^{a} b )\, e^{i b \beta}.
\end{equation}
Since we can write $F \in A(G)$ as $F = h_1 \ast (\Delta h_2^\ast)$, for $h_1, h_2 \in L^2(G, \dd\mu)$, substituting it into \eqref{eq-fabetafromF} gives the most general expression for $f \in \algA$ in terms of $h_1$ and $h_2$. This expression implies $\tau(f) = \fabeta(0,0) = 0$ directly.

For $\nu=0$, we saw in section \ref{subsec-tracefromrepresentations} that is $\pi_0(g)$ is trace-class then $\tau(g) = 0$. For $f \in \algA$, $\tr_0(f^\ast \ast f) = \int_{\gR} \dd u \, (\fabeta^\ast \astabeta \fabeta)(0,0) =  \int_{\gR} \dd u \, \tau(f^\ast \ast f) = 0$ if $\tau(f^\ast \ast f) = 0$.
\end{proof}

Notice that $(\pi_-\oplus \pi_+)(f)\in \caL^1(\caH)$ for $f \in \algA$, implies that $f \in \caK_{-} \oplus \caK_{+}$ (which is of course a stronger result than $(\pi_-\oplus \pi_+)(f)$ compact).

\begin{proposition}
\label{prop-traceclass}
For any $f \in \algA$, $\pi_\pm(f) \,\theta^2$, $\theta\, \pi_\pm(f) \,\theta$, and $\theta^2\, \pi_\pm(f)$ are trace-class operators and
\begin{equation*}
\Tr\big(\pi_\nu(f) \,\theta^2\big) = \Tr\big(\theta\, \pi_\nu(f) \,\theta\big) = \Tr\big(\theta^2\, \pi_\nu(f)\big) =
\begin{cases}
\displaystyle\int_0^\infty \dd \beta\, \fabeta(0,\beta) & \text{ for $\nu=+$},\\[10pt]
\displaystyle\int_{-\infty}^0 \dd \beta\, \fabeta(0,\beta) & \text{ for $\nu=-$}.
\end{cases}
\end{equation*}
\end{proposition}

\begin{proof}
Thanks to Proposition~\ref{eclatement}, we can replace $f$ by $g\ast h$, for $g,h \in \algA$. We consider only the case of $\pi_\pm(f) \, \theta^2$. Using Lemma~\ref{lem-commutepinudelta}, one has 
\begin{align*}
\pi_\pm(g \ast h)\,\theta^2 &= \pi_\pm(g)\, \pi_\pm(h) \,\theta^2 = \pi_\pm(g)\, \theta^2\, \pi_\pm(\Delta^{-1} h)
= \theta \, \pi_\pm(\Delta^{-1/2} g)\, \theta \,\pi_\pm \big( \Delta^{-1/2} (\Delta^{-1/2} h)\big) \\
&= \caP_\pm(g) \,\caP_\pm(\Delta^{-1/2} h)
\end{align*}
which is trace-class because $\algA \subset L^2(G, \dd\mu)$ and $\Delta^{-1/2}\algA \subset L^2(G, \dd\mu)$.

The trace is computed using the kernel $K(s,u) = \fabeta(u-s, \nu e^{-s}) e^{-u}$ of $\pi_\nu(f) \,\theta^2$:
\begin{equation*}
\Tr \big(\pi_\nu(f) \,\theta^2\big) = \int_{\gR} \dd u\, \fabeta(0, \nu e^{-u})\, e^{-u} = - \nu \, \int_{\nu \times \infty}^0 \dd \beta\, \fabeta(0,\beta)
\end{equation*}
which gives the result.
\end{proof}

\begin{corollary}
When $f\in \algA$, $\pi_\pm(\delta_2 f)$ is trace-class and
\begin{equation*}
\tau(f)=\tr_-(\delta_2 f) = -\tr_+(\delta_2 f).
\end{equation*}
\end{corollary}

For $\nu=0$, we have seen that $\pi_0(\delta_2 f) = 0$ for any $f \in \algA$, but $\tau(f) = \fabeta(0,0)$ can be non zero, so that there is no relation between these two quantities.

\begin{proof}
A direct computation in the variables $(a,\beta)$ shows that 
\begin{equation*}
\big(\pi_\pm(\delta_2 f \big) \,\phi)(s) = \int_{\gR} \dd u\, (\pm)  e^{-s} \,(\partial_\beta \fabeta)(u-s, \nu e^{-s})\,\phi(u) 
= \pm \theta^2 \,\big(\pi_\pm( \partial_\beta \fabeta) \,\phi \big)(s).
\end{equation*}
By Proposition~\ref{prop-traceclass}, the operator $\theta^2\, \pi_\pm( \partial_\beta \fabeta)$ is trace-class because $\partial_\beta \fabeta \in \algA$. The traces are computed using \eqref{eq-expressiontracesplusmoins}.
\end{proof}

%%%%%%%%%%%%%%%%%%%%%%%%%%%%%%%%%%%%%%%%%%%%%%
\section{The spectral triple}
\label{Spectral triple}
%%%%%%%%%%%%%%%%%%%%%%%%%%%%%%%%%%%%%%%%%%%%%%

%%%%%%%%%%%%%%%%%%%%%%%%%%%%%%%%%%%%%%%%%%%%%%
\subsection{About the choices}

The choice of a spectral triple $(\algA,\caH,\caD)$ is by definition a choice of a geometry. In the beginning we had just the Lie-algebra type commutation relations \eqref{commutationkappa} of the $\kappa$-deformation space to which we naturally associated the affine group $G$. From a noncommutative point of view, a natural algebra to represent this ``noncommutative space'' is the group $C^\ast$-algebra $C^\ast(G)$. Thus it is natural to take the spectral triple algebra $\algA$ as a dense subalgebra of $C^\ast(G)$. Of course this choice should be compatible with the domain of the privileged operator $\caD$ and this in turn needs a choice of $\caH$ which actually means a choice of a faithful representation $\pi$ of $\algA$ on $\caH$. However, still the crucial choice is that of an operator $\caD$.

There might be various hints, as to which ingredients could be used. For instance, consider the two generators $T$ and $X$ of the Lie algebra $\lieG \simeq \gR^2$ of $G$ given by $T \vc (1,0)$ and $X \vc (0,1)$, with respective flows $\varphi_{t, T}(a,b) = (a+t, b)$ and $\varphi_{t, X}(a,b) = (a, e^{-a} t + b)$, from which we deduce the Lie bracket
\begin{equation*}
[T,X] = - X.
\end{equation*}
Notice that $x^0 = -iT$ and $x^1 = -iX$ satisfy the relation \eqref{commutationkappa} for $n=2$ and $\kappa = 1$. Denoting the induced representation of the Lie algebra by $\dd\pi_\nu$ we have:
\begin{equation*}
(\dd\pi_\nu(T)\phi)(s) = (\partial_s \phi)(s)
\quad\text{ and }\quad
(\dd\pi_\nu(X)\phi)(s) = \nu i e^{-s} \phi(s),
\end{equation*}
so that 
\begin{equation*}
\dd\pi_\nu(T) = i \pi_\nu(\ualpha) = i \partial_2
\quad\text{ and }\quad
\dd\pi_\nu(X) = i \pi_\nu(\ubeta) = 2 \pi i \nu \theta^2.
\end{equation*}
where $\partial_2$ is defined in \eqref{eq-defpartialk}. Therefore, for any $f \in \algA$, the commutator $[\dd\pi_\nu(T), \pi_\nu(f)] = - \pi_\nu(\delta_2 f)$ is a bounded operator, and, moreover, using Lemma~\ref{lem-commutepinudelta}, one has 
\begin{equation*}
[ \dd\pi_\nu(X), \,\pi_\nu(f)] = 2 \pi i \nu [\theta^2, \,\pi_\nu(f)] = 2 \pi i \nu \,\theta^2 \,\big(\pi_\nu(f) - \pi_\nu(\sigma(f)) \big),
\end{equation*}
which, by Proposition~\ref{prop-traceclass}, is even trace-class for $\nu \in \{-,+\}$. The last commutator is closely related to the twisted commutator, which vanishes for every $f \in \algA$:
\begin{equation*}
\dd\pi_\nu(X) \,\pi_\nu(f) - \pi_\nu\big(\sigma^{-1}(f)\big) \,\dd\pi_\nu(X) = 0.
\end{equation*}
The above unbounded operators, which have bounded commutators with the algebra, are interesting candidates for a geometry of the $\kappa$-deformed space. We mention them, however, just to indicate that there are many possibilities and that the choices are not obvious.  Even if the latter approach deserves further study, we will privilege in the following the derivation-based ``noncommutative geometry'' of $C^\ast(G)$. In particular, this means that we will use the two natural derivations $\delta_1,\delta_2$ defined in section \ref{subsec-derivations} on the algebra $\algA = \caD(G)_\ast$, in order to construct a Dirac-like operator $\caD$, similarly as for the noncommutative two-torus.

\begin{definition}
The algebra, its representation on the Hilbert space and the operator $\caD$ are
\begin{align*}
& \algA \vc \caD(G)_\ast,\\
& \caH \vc \bigoplus_{\nu = +,0,-} (L^2(\gR,ds) \otimes \gC^2),\\
& \pi \vc  \bigoplus_{\nu = +,0,-} \text{Diag}(\pi_\nu,\pi_\nu),\\
& \caD \vc \gamma^k\, \partial_k \otimes \bbbone_3 \,,\\
& \chi \vc -i\gamma^1\gamma^2  \otimes \bbbone_3\,,
\end{align*}
where $\gamma^1=\left(\begin{smallmatrix} 0 &1\\ 1 & 0 \end{smallmatrix}\right)$ and $\gamma^2=\left(\begin{smallmatrix} 0 &-i\\ i &\, \,0 \end{smallmatrix}\right)$ are the Pauli matrices.
\end{definition}

Since the unbounded operators $\partial_k$ defined in \eqref{eq-defpartialk} are hermitean, the operator $\caD$ has a selfadjoint extension. The representation is chosen not only to include the two irreducible representations in the principal series, $\pi_\pm$, but also to include the other irreducible representations contained in $\pi_0$. Actually, many of the computations performed below will split on the three cases $\nu \in \{-,0,+\}$, and the main results (Theorem~\ref{dimensionspectrum} for instance) will be valid even if $\pi_0$ where not included in our representation or if at the contrary $\pi_\pm$ are excluded.

\begin{proposition}
$(\algA,\caH,\caD)$ is a even (with grading $\chi$) regular spectral triple. 
\end{proposition}

\begin{proof}
The operator $\caD$ is an unbounded selfadjoint operator on $\caH$ with the domain being the Schwartz space $\mathcal{S}(\gR)$ such that
\begin{align}
\label{Dcarre}
\caD^2= \begin{pmatrix} H+1 & 0 \\ 0 & H-1 \end{pmatrix} \otimes \bbbone_3 \,\,
\text{ where }\,\, H \vc -\tfrac{d\,}{ds}^2 +s^2.
\end{align}
Thus $H$ is the Hamiltonian of quantum harmonic oscillator with spectrum $\sigma(H) = \{2n+1 \,\mid\, n\in \gN\}$, so $\sigma(\caD^2)=\{d_n^2\vc2n\, \mid \, n\in \gN\}$ with multiplicity $m_n$ of $d_n$ equal to 1 if $n=0$ and 2 otherwise.

Since the resolvent of $\caD$ is compact, it is sufficient to show that for any $f \in \algA$, $[\caD,\pi(f)]$ is a bounded operator on $\caH$. By direct computation,
\begin{align}
[\caD,\pi(f)]= \bigoplus_{\nu = +,0,-} A_\nu, \qquad
A_\nu \vc \gamma^k \pi_\nu(i\delta_k f) \label{[D,.]}
\end{align}
because $[\partial_k,\pi_\nu(f)]=\pi_\nu(i\delta_k f)$ for $k\in \{1,2\}, \,\nu \in \{-,+,0\}$ and $\pi_0(\delta_2 f)=0$. Since $\pi_\nu(g)$ is bounded for any $g \in \algA$, the claim is proved.

Regularity of the triple means that $\algA$ and $[\caD,\pi(\algA)]$ are in $\cap_{n=0}^\infty \text{ dom }\delta^n$ when $\delta\vc \text{ad}(\abs*{\caD})$. Let $L(T) \vc \langle \caD \rangle^{-1}[\caD^2,T]$ and $R(T) \vc [\caD^2,T] \langle \caD \rangle^{-1}$, where $\langle \caD \rangle \vc (1+\caD)^{1/2}$. Then, applying \cite[p.~238]{ConnMosc95a}, \cite[Lemma~10.2.3]{GracVariFigu01a} and \cite{CareGayrRenn12a}, we have to show that for any $T\in \pi(\algA) \cup [\caD,\pi(\algA)]$, $R^m(T) \circ L^n(T)=\langle \caD \rangle^{-n}\big(\text{ad}(\caD^2)\big)^{n+m}(T)\,\langle \caD \rangle^{-m}$ is a bounded operator.

Since $\pi$ is diagonal, we may by restriction assume that $\pi=\pi_\nu\,\bbbone_2$ for some $\nu =-,0,+$, and $\caD=\gamma^k \partial_k$, so \eqref{[D,.]} says that the case $T\in [\caD,\algA]$ reduces to the case $T\in \algA$ since the $\gamma$-matrices are bounded. Thus we assume $T=\pi(f)$ and we have to show that 
\begin{equation*}
\langle \gamma^k \partial_k \rangle^{-n}\big(\text{ad}((\gamma^k \partial_k)^2)\big)^{n+m}(\pi_\nu(f)\bbbone_2)\,\langle \gamma^k \partial_k \rangle^{-m}
\end{equation*}
is bounded. First, observe that 
\begin{align*}
\text{ad} (\caD^2 )\,\pi(f)=[ (\gamma^k \partial_k)^2,\pi(f)]
=[\partial_1^2+\partial_2^2,\pi_\nu(f)]\,\bbbone_2
= \Big(-\pi_\nu\big((\delta_1^2+\delta_2^2)\fabeta \big)+2\pi_\nu(\delta_k f)\partial_k\Big)\,\bbbone_2.
\end{align*}
By iterative application of this formula, taking into account nontrivial commutations between $\partial_k$, we obtain an expression, which is a polynomial of degree at most $m+n$ in the operators $\partial_k$ with coefficients from $\pi_\nu(\algA)$. So, to end the proof it is sufficient to show that for any $a \in \algA$, any $n, m, n_1,n_2 \geq 0$ such that $n+m=n_1+n_2$ the operator
\begin{equation*}
\langle \gamma^k \partial_k \rangle^{-m} \,\pi_n(a) \partial_1^{n_1} \partial_2^{n_2}
\,\langle \gamma^k \partial_k \rangle^{-n},
\end{equation*}
is bounded, which can be done similarly as in \cite[Corollary~5]{GayrWulk11a}.

One checks that the triple is even since $\chi=\chi^\ast$, $[\chi,\pi(f)]=0$ for $f\in \algA$ and $\caD \chi=-\chi \caD$.
\end{proof}

%%%%%%%%%%%%%%%%%%%%%%%%%%%%%%%%%%%%%%%%%%%%%%
\subsection{Metric dimension}

The zeta-function of $\caD$ is defined by
\begin{align*}
\zeta_{\caD}(s) \vc \Tr \big( (1+\caD^2)^{-s/2}\big) \text{ for }s \in \gC \text{ with }\Re(s) \text{ large enough}.
\end{align*}

\begin{lemma}
If $\zeta$ is the standard Riemann zeta-function, then
\begin{align*}
\zeta_{\caD}(s) = (2-2^{1-s/2}) \,\zeta(\tfrac{s}{2})-1
\end{align*}
has a unique pole at $s=2$. Thus the metric dimension of the triple $(\algA,\caH,\caD)$, defined as the infimum of all  $s\in \gR^*$ such that $\Tr \big( (1+\caD)^{-s/2})\big)<\infty$, is $2$. 
\end{lemma}

\begin{proof}
\begin{align*}
\zeta_{\caD}(2s) = \sum_{n=0}^\infty \tfrac{m_n}{(d_n^2+1)^s}=1+ 2\sum_{n=1}^\infty \tfrac{1}{(2n+1)^s}
= 1+2^{1-s} \sum_{n=1}^\infty \tfrac{1}{(n+1/2)^s}= 1+2^{1-s} \big(\zeta_{1/2}(s) -2^s\big)
\end{align*}
where the Hurwitz zeta-function $\zeta_{1/2}(s)\vc \sum_{n=0}^\infty \tfrac{1}{(n+1/2)^s}$ satisfies $\zeta_{1/2}(s)=(2^s-1)\zeta(s)$. Thus $\zeta_\caD(2s)=(2-2^{1-s})\zeta(s)-1$, the unique pole of $\zeta_\caD$ is at  $s=2$ and this pole is simple.
\end{proof}

%%%%%%%%%%%%%%%%%%%%%%%%%%%%%%%%%%%%%%%%%%%%%%
\subsection{Spectral dimension and Dixmier trace}
\label{subsec-metricdim}

If we want to find the spectral dimension  $d \in \gR_+$ of $(\algA,\caH,\caD)$, which is a nonunital spectral triple, we should rather use the following definition \cite[Definition 6.1]{CareGayrRenn12a}
\begin{equation}
\label{eq-defspectraldimension}
d\vc\inf\{d'>0\,\mid\, \Tr\big(\pi(f)\,(1+\caD^2)^{-d'/2}\big)<\infty \text{ for any }f\in \algA^+\}.
\end{equation}
Actually, we shall prove something stronger: $(\algA,\caH,\caD)$ is $\mathcal{Z}_1$-summable (see \cite[Definition 6.2]{CareGayrRenn12a}), namely
\begin{align*}
\underset{s \,\downarrow \,1}{\limsup} \, \,\abs*{(s-1)\,\Tr\big(\pi(f)\,(1+\caD^2)^{-s/2}\big)} <\infty \text{ for all } 
f \in \algA.
\end{align*}
This means that $\pi(f)\,(1+\caD^2)^{-d/2} \in \caL^{1,\infty}(\caH)$ or that its $m$-th singular value behaves  like $\mathcal{O}(m^{-1})$.  

\begin{theorem}
\label{spectraldim}
The spectral triple $(\algA,\caH,\caD)$ is $\mathcal{Z}_1$-summable (thus its spectral dimension is $d=1$). Moreover, for any $f \in \algA$, the operators $\pi(f)\,(1+\caD^2)^{-1/2}$ are measurable and for any Dixmier trace $\Trw$, we have 
\begin{align}
\Trw \big[\pi(f)\,(1+\caD^2)^{-1/2}\big] 
= \sum_{\nu \in \{-,0,+\}} \Trw \big[\pi_\nu(f)\,(1+\caD^2)^{-1/2}\big]
= 8\,\tau(f). \label{result}
\end{align}
In particular, $\Trw \big[\pi(f\ast f^\ast)\,(1+\caD^2)^{-d/2}\big] = 8\,\int_\gR  da\,|\fabeta(a,0)|^2$. 
\end{theorem}

Notice that if $\pi(f)$ is trace-class, then by Proposition~\ref{trace-class} all terms of \eqref{result} vanish.

\begin{theorem}
\label{dimensionspectrum}
The dimension spectrum (see \cite{GayrWulk11a}) of $(\algA,\caH,\caD)$ is $\{1-\gN\}$.
\end{theorem}

To get these results, and in particular the measurability which is the independence of the result in the choice of a Dixmier trace $\Trw$, we need to know the behavior around $t=0$ of $\Tr\big(\pi(f)\,e^{-t\caD^2}\big)$ and we will follow closely \cite{GayrWulk11a} which also use an harmonic-like operator $\caD$.

The operator $e^{-t\caD^2}$ is trace-class and 
\begin{align}
\Tr(e^{-t\caD^2})=6\Tr(e^{-tH})\,\cosh t=3\coth t. \label{eq-heattraceDirac}
\end{align}
The heat trace associated to $\caD$ is, using \eqref{Dcarre},
\begin{align}
\Tr(\pi(f)\,e^{-t\caD^2}) = 2\cosh(t)\,\sum_\nu \Tr \big(\pi_\nu(f)\,e^{-tH}\big),\, \quad t\in \gR^+.
\end{align}

\begin{lemma}
Let $f\in \algA$ and $\nu \in \{-,0,+\}$. Then
\begin{align}
\Tr\big( \pi_\nu(f)\,e^{-tH}\big)=\tfrac{1}{2\pi \sqrt{\cosh 2t}} \int_{\hgR\times\gR} \dd v \,\dd x\, \falpbeta(v,\nu e^{-x}) \, e^{-\tfrac{1}{2} (\tanh 2t)(x^2+v^2)\, -\,i \,\tfrac{2\sinh^2 t}{\cosh^2 t +\sinh^2t} \,xv}. \label{heattrace}
\end{align}
In particular, when $\nu=0$, 
$\Tr\big(\pi_0(f)\,e^{-tH} \big)= \tfrac{1}{\sqrt{2\pi \sinh 2t}} \,\int_{\hgR} \dd v\, \falpbeta(v,0)\,e^{-(\tanh t)\,v^2}$.
\end{lemma}

\begin{proof}
It is known that the kernel of the heat operator $e^{-tH}$ is given by Mehler's formula \cite{BerlGetzVerg03a}
\begin{align*}
K_{e^{-tH}}(x,y)= \tfrac{1}{\sqrt{2\pi \sinh(2t)}}\, e^{-\tfrac{1}{4}[\coth(t)\,(x-y)^2+\tanh(t)\,(x+y)^2]}.
\end{align*}
In the variables $(\alpha,\beta)$, \eqref{eq-Schwartzkernelalpbeta} gives
$K_{\pi_\nu(f)}(x,y) = \tfrac{1}{2\pi} \int_{\gR} \dd v\,\falpbeta(v,\nu e^{-x})\,e^{-iv(y-x)}$, so that
\begin{align*}
K_{\pi_\nu(f)\,e^{-tH}}(x,y)
=\tfrac{1}{2\pi}\int_{\gR\times\hgR} \dd u\,\dd v \falpbeta(v,\nu e^{-x}) \,e^{-iv(u-x)}\,\tfrac{1}{\sqrt{2\pi \sinh 2t}}\, e^{-\tfrac{1}{4}[(\coth t)\,(u-y)^2+(\tanh t)\,(u+y)^2]}.
\end{align*}
The integration along $u$ in this expression can be performed, and one has:
\begin{align*}
\int_{\gR} \dd u \,e^{-iv(u-x)}\,e^{-\tfrac{1}{4}[\coth(t)\,(u-y)^2+\tanh(t)\,(u+y)^2]} =
\sqrt{2\pi \tanh 2t}\,e^{-\tfrac{1}{2} (\tanh 2t)(x^2+y^2)\,-\,i\, \tfrac{2\sinh^2 t}{\cosh^2 t +\sinh^2 t} \,xv}.
\end{align*}
Since $e^{-tH}$ is trace-class by \eqref{eq-heattraceDirac}, so is $\pi_\nu(f) e^{-tH}$, thus 
\begin{align*}
\Tr\big( \pi_\nu(f)\,e^{-tH}\big)
&=\int_{\gR} \dd x\,K_{\pi_\nu(f)\,e^{-tH}}(x,x)
\\
&=\tfrac{1}{2\pi \sqrt{\cosh 2t}} \int_{\hgR\times\gR} \dd v \,\dd x\, \falpbeta(v,\nu e^{-x}) \, e^{-\tfrac{1}{2} (\tanh 2t)(x^2+v^2)\, -\,i \,\tfrac{2\sinh^2 t}{\cosh^2 t +\sinh^2t} \,xv}. 
\end{align*}
Moreover
\begin{align*}
\int_{\gR} \dd x \,e^{-\tfrac{1}{2} (\tanh 2t)(x^2+y^2)\,
-\,i \,\tfrac{2\sinh^2 t}{\cosh^2 t +\sinh^2 t} \,xv}=\sqrt{\pi}\sqrt{\coth t+\tanh t} \, \,e^{-(\tanh t)\,v^2}
\end{align*}
so the result for $\Tr\big(\pi_0(f)\,e^{-tH} \big)$ follows directly after 
performing integration in $x$.
\end{proof}

In the following computations, we use the two functions defined on $t>0$:
\begin{align*}
T(t) \vc \tfrac{1}{2}\tanh 2t,
\qquad
S(t) \vc \tfrac{2\sinh^2 t}{\cosh^2 t +\sinh^2 t}
\end{align*}
sometimes also denoted by $T$ and $S$ in computations.

\begin{lemma}
\label{lem-Ifortsmall}
For any $f\in \algA$, let
\begin{equation}
\label{I(t)}
I_\nu(t) \vc \tfrac{\sqrt{\pi}}{\sqrt{T(t)}} \int_{\gR^2} \dd x \dd y \,\fabeta(y,\nu e^{-x})\,e^{-T(t)x^2}\,e^{-\tfrac{1}{4T(t)}(S(t)x-y)^2}.
\end{equation}
Then 
\begin{equation*}
\lim_{t \to 0+} \sqrt{t} \,I_\nu (t) = 
\begin{cases}
\pi^{\tfrac{3}{2}} \fabeta(0,0) = \pi^{\tfrac{3}{2}} \tau(f) & \text{ for $\nu\in \{-,+\}$},\\[10pt]
2 \pi^{\tfrac{3}{2}} \fabeta(0,0) = 2 \pi^{\tfrac{3}{2}} \tau(f) & \text{ for $\nu=0$}.
\end{cases}
\end{equation*}
\end{lemma}

\begin{proof}
In the integral \eqref{I(t)}, we make the change of variables $x' = \sqrt{T} x$ and $y' = \tfrac{y - S x}{\sqrt{T}}$, so that
\begin{equation*}
I_\nu(t) = \tfrac{\sqrt{\pi}}{\sqrt{T}} \int_{\gR^2} \dd x' \dd y' \,\fabeta(\sqrt{T} y' + \tfrac{S}{\sqrt{T}} x',\nu e^{-x'/\sqrt{T}})\,e^{-x'^2}\,e^{-\tfrac{1}{4}y'^2}.
\end{equation*}
One has 
\begin{equation*}
\lim_{t \to 0+} \fabeta(\sqrt{T} y' + \tfrac{S}{\sqrt{T}} x',\nu e^{-x'/\sqrt{T}}) =
\begin{cases}
\fabeta(0,0) & \text{ if $x'>0$},\\
\lim_{\beta \to \nu \times \infty} \fabeta(0, \beta) = 0 & \text{ if $x' < 0$}.
\end{cases}
\end{equation*}
The function $\abs{\fabeta}$ is bounded, so that the modulus of the integrand is dominated by a constant times the two Gaussian functions, which is integrable, so that we can apply the dominated convergence theorem.
For $\nu=0$, we get 
\begin{equation*}
\lim_{t\to 0+} \sqrt{t}\, I_\nu(t) = \sqrt{\pi} \fabeta(0,0) \int_{-\infty}^{\infty} \dd x'\, e^{-x'^2}\, \int_{-\infty}^{\infty} \dd y' \,\,e^{-y'^2/4} = 2 \pi \sqrt{\pi} \fabeta(0,0).
\end{equation*}
For $\nu \in \{-,+\}$, we split the integral along $x'$ into $x'>0$ and $x'<0$, and we get
\begin{equation*}
\lim_{t\to 0+} \sqrt{t}\, I_\nu(t) = \sqrt{\pi} \fabeta(0,0) \int_{0}^{\infty} \dd x'\, e^{-x'^2}\, \int_{-\infty}^{\infty} \dd y' \,\,e^{-y'^2/4} = \pi \sqrt{\pi} \fabeta(0,0).
\end{equation*}
\end{proof}

\begin{lemma}
For any $f\in \algA$, there exists a constant $c(f)$ such that
\begin{align}
\label{tracesym}
\Tr\big(\pi(f^\ast)\,e^{-t\caD^2}\,\pi(f)\big) \leq c(f)\, \max(1,\tfrac{1}{\sqrt{t}})\,.
\end{align}
\end{lemma}

\begin{proof}
Since the function in (\ref{tracesym}) is continuous in $t>0$, and has a limit when multiplied by $\sqrt{t}$
at $t=0$, the result could be deduced from continuity. However, since we shall conjecture that similar result appears in a more general situation, we prefer to give a computable proof.
\\
The equation \eqref{heattrace} entails
\begin{align}
\Tr\big(\pi(f^\ast)\,e^{-t\caD^2}\,\pi(f)\big) 
&=\Tr\big(\pi(f\ast f^\ast)\,e^{-t\caD^2}\big)
\nonumber
\\
&\hspace{-1.5cm}= \tfrac{2\cosh t}{2\pi \sqrt{\cosh 2t}} \, \sum_\nu \int_{\gR^2} \dd v \dd x\, (\falpbeta\astalpbeta\falpbeta^\ast)(x,\nu e^{-x}) \,e^{-T (x^2+y^2)-iS xv} 
\label{eq-fabfabexpD2}
\\
&\hspace{-1.5cm}= \tfrac{\cosh t}{\pi \sqrt{\cosh 2t}} \sum_\nu \tfrac{\sqrt{\pi}}{\sqrt{T}} \int_{\gR^2} \dd x \dd y \,(\fabeta\astabeta \fabeta^\ast)(y,\nu e^{-x}) \,e^{-T\, x^2 -\tfrac{1}{4T}(Sx-y)^2}
\label{eq-fabetafabetaexpD2}
\\
& \hspace{-1.5cm}\leq \tfrac{\sqrt{t}} {\sqrt{t}}\left[ \tfrac{\cosh t}{\pi \sqrt{\cosh 2t}} \sum_\nu \sqrt{\tfrac{\pi}{T(t)}} \int_{\gR^2} \dd x \dd y\, \abs*{\fabeta\astabeta \fabeta^\ast}(y,\nu e^{-x}) \,e^{-T(t)\, x^2 -\tfrac{1}{4T(t)}(S(t)x-y)^2} \right].
\label{bracket}
\end{align}
To prove \eqref{eq-fabetafabetaexpD2}, notice that, with $g \vc  f \ast f^\ast$, the integral in $v$ in \eqref{eq-fabfabexpD2} is equal to
\begin{align}
\int_{\hgR} \dd v \, \galpbeta(v,\nu e^{-x}) \, e^{-Tv^2}\,e^{-\,i S \,xv}
&=\tfrac{\sqrt{\pi}S}{\sqrt{T}}\, \int_{\gR} \dd y\, \gabeta(Sy,\nu\,e^{-x})\,e^{-\tfrac{S^2}{4T}\,(x-y)^2} \label{change}
\\
&=\tfrac{\sqrt{\pi}}{\sqrt{T}} \int_{\gR} \dd y' \,\gabeta(y',\nu e^{-x})\,e^{-\tfrac{1}{4T}(Sx - y')^2}. \label{change1}
\end{align}
Actually the right hand side of \eqref{change} is equal to
$\tfrac{\sqrt{\pi}S}{\sqrt{T}} \,\tfrac{1}{2\pi}\,\int_{\gR\times\hgR} \dd y\,\dd v \,\galpbeta(v,\nu\,e^{-x})\,e^{-iSyv} \,e^{-\tfrac{S^2}{4T}(x-y)^2}$ 
and the integration in $y$ gives $\int_{\gR} \dd y\, e^{-iSyv} \,e^{-\tfrac{S^2}{4T}(x-y)^2}= 2\tfrac{\sqrt{\pi T}}{S} \,e^{-Tv^2}\,e^{-iSxv}$ which proves \eqref{change} and \eqref{change1} after a change of variable. 

Denote by $A(t)$ the expression in the bracket of \eqref{bracket}. When $t \to 0$, Lemma~\ref{lem-Ifortsmall} implies that $\sqrt{t}\, A(t)$ goes to a constant which depends on $\fabeta(0,0)$. Choosing a constant $c_1(f)$ sufficiently large, there exist $0 < t_0 <1$ such that for all $t \in (0,t_0)$, $\sqrt{t}\, A(t) \leq c_1(f)$.

The function $(x,y) \mapsto \abs*{\fabeta\astabeta \fabeta^\ast}(y,\nu e^{-x})$ is compactly supported in $y$, with support included in $[-m,m]$, and bounded in $x$ and $y$, so that there is a constant $c_2(f)$ (which also depends on $t_0$) such that, for all $t\geq t_0$, 
\begin{equation*}
A(t) \leq \tfrac{\cosh t}{\pi \sqrt{\cosh 2t}}  \tfrac{6m\sqrt{\pi}}{\sqrt{T(t)}}\, \norm{\fabeta}_\infty \int_{\gR} \dd x\, e^{-T(t) \,x^2} \leq c_2(f).
\end{equation*}
\\
With $c(f) = \max\left(\tfrac{c_1(f)}{\sqrt{t_0}}, c_2(f)\right)$, one gets \eqref{tracesym}.
\end{proof}

\begin{lemma}
\label{lem-normtracecommutatorexpD2}
There exists a constant $C$ such that for any $f\in\algA$, 
\begin{align}
\norm*{\left[\pi(f),e^{-t\caD^2}\right]}_1 \leq C \sqrt{t} \sum_k 
\norm*{\pi \left( \delta_k(f) \right) \,e^{-\frac{t}{4} \,\caD^2}}_1\,.
\label{comm1}
\end{align}
\end{lemma}

\begin{proof}
Since $[e^A,B]=\int_0^1 \dd s\, e^{sA}\,[A,B]\,e^{(1-s)A}$, one has 
\begin{equation*}
[\pi(f), e^{-t\caD^2}]=-t\int_0^1 \dd s\, e^{-ts\caD^2} \,[\caD^2,\pi(f)]\, e^{-t(1-s)\caD^2}.
\end{equation*}
 
Moreover, $e^{-t\caD^2}$ is trace-class, so writing the commutator 
$[\caD^2,\pi(f)]= \caD [\caD,\pi(f)] + [\caD,\pi(f)] \caD$, we get
\begin{align*}
\norm*{[\pi(f),e^{-t\caD^2}]}_1 \leq t \int_0^1 \dd s\, \norm*{e^{-ts/2 \,\caD^2} \caD} \, \norm*{e^{-ts/2\,\caD^2}[\caD,\pi(f)]\,e^{-t(1-s)/2\,\caD^2}}_1\, \norm*{e^{-t(1-s)/2\,\caD^2}} \\
+ \norm*{e^{-ts/2\,\caD^2}}\, \norm*{e^{-ts/2\,\caD^2}[\caD,\pi(f)]\,e^{-t(1-s)/2\,\caD^2}}_1\, \norm*{\caD e^{-t(1-s)/2\,\caD^2}}.
\end{align*}
Since $\norm*{\caD\,e^{-t\caD^2}}=\underset{x\in \gR}{\sup} \,\,\abs*{xe^{-tx^2}}=\tfrac{1}{\sqrt{2et}}= \vcentcolon c\tfrac{1}{\sqrt{t}}\,$, 
\begin{align*}
\norm*{[\pi(f),e^{-t\caD^2}]}_1 & \leq c\, t \int_0^1 \dd s \, \norm*{e^{-ts/2\,\caD^2} [\caD,\pi(f)]\,e^{-t(1-s)/2\,\caD^2}}_1 
(\tfrac{1}{\sqrt{ts}}+\tfrac{1}{\sqrt{t(1-s)}})\\
&=c\,\sqrt{t} \int_0^1 \dd s\,(\tfrac{1}{\sqrt{s}}+\tfrac{1}{\sqrt{1-s}}) \norm*{e^{-ts/2\,\caD^2}[\caD,\pi(f)]\,e^{-t(1-s)/2\,\caD^2}}_1\,.
\end{align*}
We claim that
\begin{align*}
\norm*{e^{-ts/2\,\caD^2}[\caD,\pi(f)]\,e^{-t(1-s)/2\,\caD^2}}_1 \leq 
 \begin{cases}
\sum_k \norm*{\pi(\delta_k(f))\,e^{-t/4\,\caD^2}}_1 &\text{ if $s \in [0,\frac{1}{2}]$},\\[6pt]
\sum_k \norm*{e^{-t/4\,\caD^2}\,\pi(\delta_k(f))}_1 &\text{ if $s \in [\frac{1}{2},1]$}.
\end{cases}
\end{align*}
This yields the result with 
$C\vc 3c\int_0^1 ds\,(\tfrac{1}{\sqrt{s}}+\tfrac{1}{\sqrt{1-s}})
=\tfrac{12}{\sqrt{2 e}}$. \\
Proof of the claim: for $\nu \in \{-,0,+\}$, let $B_\nu \vc e^{-ts/2 \,(\gamma^k \partial_k)^2}\,\gamma^p \pi_\nu(\delta_p(f)) \,e^{-t(1-s)/2 \,(\gamma^k \partial_k)^2}$, then
\begin{align*}
\norm*{\text{Diag}(B_-,B_+,B_0)}_1 & = 
\sum_\nu \norm*{B_\nu}_1  \leq \sum_\nu \sum_p \norm*{\gamma^p}\,
\norm*{e^{-ts/2 (\gamma^k \partial_k)^2}\,\pi_\nu\big(\delta_p(f)\big)\,
      e^{-t(1-s)/2 (\gamma^k \partial_k)^2}}_1\\
&=\sum_\nu \sum_p \norm*{e^{-ts/2\,\caD^2}\,
  \pi_\nu\big(\delta_k(f)\big)\, e^{-t(1-s)/2\,\caD^2}}_1\\
&\leq  \sum_p \norm*{e^{-ts/2 \,\caD^2}}\,
\norm*{\pi\big(\delta_p(f)\big)\,e^{-t/4\,\caD^2}}_1 \quad \text{for } 
s\in [0,\frac{1}{2}]
\end{align*}
since $\norm*{XY}_1\leq \norm*{XZ}_1$ for $0\leq Y\leq Z$. Case 
$s\in [\frac{1}{2},1]$ is similar.
\end{proof}

\begin{lemma}
\label{lem-normtracepiexpD2}
For any $f\in \algA$, there exists a constant $C(f)$ such that 
\begin{equation*}
\norm*{\pi(f)\,e^{-t\caD^2}}_1\leq C(f) \max \left(\tfrac{1}{\sqrt{t}}, \sqrt{t}\right).
\end{equation*}
\end{lemma}

\begin{proof}
By Proposition~\ref{eclatement}, we may assume that $f=g\ast h$.\\ 
Since $\pi(f)\,e^{-t\caD^2}=\pi(g)\,e^{-t\caD^2}\,\pi(h)+\pi(g)[\pi(h),e^{-t\caD^2}]$, we get by \eqref{tracesym}, \eqref{comm1} and \eqref{eq-heattraceDirac},
\begin{align}
\norm*{\pi(f)\,e^{-t\caD^2}}_1 
& \leq \norm*{\pi(g)\,e^{-t/2\,\caD^2}}_2\,\norm*{e^{-t/2\,\caD^2}\,\pi(h)}_2 + \norm*{\pi(g)\,[\pi(h),e^{-t\caD^2}]}_1 \nonumber
\\
& \leq \sqrt{c(g)\,c(h)} \,\max(1,\tfrac{1}{\sqrt{t}}) + \norm*{\pi(g)}\, C\sum_k \norm*{\delta_k(h)}\,\sqrt{t} \coth t/4\,.
\label{eq-majorationtracepikernel}
\end{align}
When $t\to 0$, one has $\sqrt{t} \coth t/4 \sim \tfrac{4}{\sqrt{t}}$, and when $t \to \infty$, one has $\sqrt{t} \coth t/4 \sim \sqrt{t}$. Collecting all these asymptotic behaviors, one gets the result.
\end{proof}

\begin{lemma}
\label{commut}
For any $f\in \algA$,
\begin{align}
\label{com}
[(1+\caD^2)^{-1/2},\pi(f)] \in \caL^1(\caH).
\end{align}
\end{lemma}

\begin{proof}
Assume again that $f=g\ast h$. Since 
\begin{equation*}
[(1+\caD^2)^{-1/2},\pi(f)] = \pi(g)\,[(1+\caD^2)^{-1/2},\pi(h)] - \big(\pi(h^\ast)\;[(1+\caD^2)^{-1/2},\pi(g^\ast)] \big)^\ast,
\end{equation*}
 it is sufficient to prove $\pi(g)\,[(1+\caD^2)^{-1/2},\pi(h)] \in \caL^1(\caH)$ for any $g,h\in \algA$. Thus using Lemma~\ref{lem-normtracecommutatorexpD2} and then Lemma~\ref{lem-normtracepiexpD2}:
\begin{align*}
\norm*{\pi(g)\,[(1+\caD^2)^{-1/2},\pi(h)]}_1 
& \leq \tfrac{1}{\Gamma(1/2)}\int_0^\infty \dd t\,t^{1/2-1}\, \norm*{\pi(g)} \,\norm*{[e^{-t(1+\caD^2)},\pi(h)]}_1 \\
& \leq \tfrac{C\, \norm*{\pi(g)}}{\Gamma(1/2)} \int_0^\infty \dd t\,t^{-1/2} \,e^{-t} \sqrt{t} \sum_k \norm*{\pi(\delta_k(h)\,e^{-t/4\,\caD^2}}_1 \\
& \leq \tfrac{C \,\norm*{\pi(g)}\,\sum_k C(\delta_k(h))}{\Gamma(1/2)}\,   \int_0^\infty \dd t\, e^{-t} \max \left(\tfrac{1}{\sqrt{t}}, \sqrt{t}\right)
 \\
& < \infty.
\end{align*}
\end{proof}

\begin{proposition}
\label{Dixtr}
$\pi(f) \,(1+\caD^2)^{-1/2} \in \caL^{1,\infty}(\caH)$ for any $f \in \algA.$
\end{proposition}

\begin{proof}
First assume that
\begin{align}
\label{sup}
\underset{1\leq s\leq 2}{\sup}\,(s-1)\,\Tr\big(\pi(f^\ast)\,(1+\caD^2)^{-s/2}\,\pi(f) \big) < \infty, \text{ for any } f\in \algA.
\end{align}
This shows that $\pi(f)$ is in the $\ast$-algebra $B_\zeta\big((1+\caD^2)^{-1/2}\big)$ of \cite[Definition 1]{CareGayrRenn12a}. 
Thus by \cite[Proposition 3.8]{CareGayrRenn12a}, $\pi(f^\ast)\,(1+\caD^2)^{-s/2}\,\pi(f)\in \caL^{1,\infty}(\caH)$ and by 
polarization, $\pi(f)\,(1+\caD^2)^{-s/2}\,\pi(g)\in \caL^{1,\infty}(\caH)$ for any $f,g \in \algA$.
\\
If $f=g\ast h$, the relation 
$\pi(f)\,(1+\caD^2)^{-1/2}=\pi(g)\,(1+\caD^2)^{-1/2}\,\pi(h) - \pi(g)\,[(1+\caD^2)^{-1/2},\,\pi(h)]$ yields the result because 
by \eqref{com}, the last commutator is in $\caL^{1,\infty}(\caH)$.\\
Let us now prove \eqref{sup} using \eqref{tracesym}. We get for $1 \leq s \leq 2$,
\begin{align}
\Tr\big(\pi(f^*)\,(1+\caD^2)^{-s/2}\,\pi(f) \big)
&=\tfrac{1}{\Gamma(s/2)} \int_0^\infty \dd t\, t^{s/2-1}\,e^{-t} \Tr\big(\pi(f^\ast)\,e^{-tD^2}\,\pi(f) \big)
\nonumber
\\
& \leq \tfrac{c(f)}{\Gamma(s/2)} \int_0^\infty \dd t\, t^{s/2-1} \,e^{-t}\,\max(1,\tfrac{1}{\sqrt{t}}) \leq c(f) \tfrac{\Gamma((s-1)/2)}{\Gamma(s/2)}
\label{eq-inequalitygammafunction}
\end{align}
since for any $s \geq 1$, $\int_0^\infty \dd t\, t^{s/2-1} \,e^{-t} = \Gamma(\tfrac{s}{2})$ and $\int_0^\infty \dd t\, t^{s/2-1}  \,e^{-t}\,\tfrac{1}{\sqrt{t}} = \Gamma(\tfrac{s-1}{2})$. Thus the left hand side of \eqref{sup} is bounded by $\Gamma(\frac{1}{2})\,c(f)$.
\end{proof}

\begin{proof}[Proof of Theorem \ref{spectraldim}]
Thanks to Proposition~\ref{Dixtr}, we first claim that the computation of the Dixmier trace of $\pi(f)(1+\caD^2)^{-d/2}$ can be obtained by
\begin{align}
\Trw\big[\pi(f)\,(1+\caD^2)^{-1/2}\big]=\underset{s\downarrow 1}{\lim}\, (s-1)\,
\Tr\big[\pi(f)(1+\caD^2)^{-s/2}\big]. \label{Dix}
\end{align}

We shall prove it for $f=g\ast g^\ast$. This condition is not a restriction since any $f \in \algA$ is a finite sum of elements $g \ast h^\ast$ with $g, h \in \algA$ thanks to Proposition~\ref{eclatement} and \eqref{eq-polarization}.

Actually, for any $g \in \algA$, $a=\pi(g \ast g^\ast)\geq 0$ is such that $a^{1-1/2}=\vert \pi(g^*)\vert \in B_\zeta\big((1+\caD^2)^{-1/2})$, as seen in the proof of Proposition~\ref{Dixtr} since $a\in B_\zeta$ yields $\vert a^* \vert \in B_\zeta$ by \cite[Lemma 3.7, iii)]{CareGayrRenn12a}. Moreover, $[a^{1/2},\big((1+\caD^2)^{-1/2})]\in \caL^1(\caH) \subset \mathcal{Z}_1^0$, by Lemma~\ref{commut}. We can now apply \cite[Proposition 4.12]{CareGayrRenn12a} to prove 
\eqref{Dix}. \\
Using \eqref{heattrace}, \eqref{I(t)} and \eqref{change1},
\begin{align*}
\Tr\big[\pi(f)(1+\caD^2)^{-s/2}\big]
& =\tfrac{1}{\Gamma(s/2)}  \int_0^\infty \dd t \,t^{\tfrac{s}{2} -1}\,\sum_\nu2\cosh(t)\,\Tr\big(\pi_\nu(f)\,e^{-tH}\big) 
\\
& =\tfrac{1}{\Gamma(s/2)} \int_0^\infty \dd t \,t^{\tfrac{s}{2} -1}\,t^{-1/2}\,e^{-t}\,F(t)
\end{align*}
where
\begin{align*}
F(t) \vc \sum_\nu \tfrac{\cosh(t)\,t^{1/2}\,e^t}{\pi \sqrt{\cosh 2t}} I_\nu(t). \label{F(t)}
\end{align*}

The main interest in this formula comes from Lemma~\ref{F(0)}, so we want to know the behavior of $F(t)$ when $t \to 0$:
\begin{equation}
F(t) \underset{t\downarrow 0}{\sim} \tfrac{1}{\pi}\sum_\nu\,  \underset{t\downarrow 0}{\lim}\,\,t^{\frac{1}{2}} I_\nu(t).
\end{equation}
Lemma~\ref{lem-Ifortsmall} shows that the behavior around $t=0$ of $I_\nu(t)$ is in $\tfrac{1}{\sqrt{t}}$, so gathering \eqref{Dix} and Lemmas~\ref{F(0)}, \ref{lem-Ifortsmall}, we obtain that $\Trw\big[\pi(f)\,(1+\caD^2)^{\frac{s}{2}}\big]$ is non-zero for $s=1$ and takes the claimed value. The last equality of the theorem is \eqref{eq-tau(f*f)}.
\end{proof}

\begin{lemma}
\label{F(0)}
If $F$ is an analytic function in $t$,
\begin{align*}
\underset{s\downarrow 1}{\lim}\, \tfrac{s-1}{\Gamma(ds/2)}\int_0^\infty \dd t \,t^{\tfrac{ds}{2} -1}\,t^{-d/2}\,e^{-t}\,F(t)
= \tfrac{2}{d\,\Gamma(d/2)}\,F(0).
\end{align*}
\end{lemma}
\begin{proof}
This equality comes from
\begin{align*}
\underset{s\downarrow 1}{\lim}\, \tfrac{s-1}{\Gamma(ds/2)}\int_0^\infty \dd t \,t^{\tfrac{ds}{2} -1}\,t^{-d/2}\,e^{-t}\, t^{n}
= \tfrac{2}{d\,\Gamma(d/2)}\, \delta_{n,0}\,,\text{ for } n\in \gN,
\end{align*}
and this follows from $\Gamma(z)=\int_0^\infty \dd t\,t^{z-1}\,e^{-t}$ and 
$\underset{s\downarrow 1}{\lim}\,(s-1)\,\Gamma(\tfrac{d(s-1)+2n}{2})=\tfrac{2}{d}\,\delta_{n,0}$.
\end{proof}

\begin{proof}[Proof of Theorem \ref{dimensionspectrum}] 
Since all behavior in $t$ obtained in Lemmas~\ref{lem-Ifortsmall}--\ref{lem-normtracepiexpD2} are the same as in \cite{GayrWulk11a}, it is sufficient to follow the arguments of \cite[Theorem 17]{GayrWulk11a}.
\end{proof}

\begin{remark}
\label{remark-spectraldim}
The computation of the spectral dimension in Theorem~\ref{spectraldim} relies on the behavior at $t\to 0$ of $I_\nu(t)$ in \eqref{I(t)}, which is  $I_\nu(t) \,\underset{t\downarrow 0}{\sim}\, \tfrac{C}{\sqrt{t}}$. This behavior is used to prove inequality~\eqref{tracesym}, then the right hand side of this inequality is used in the first term of \eqref{eq-majorationtracepikernel}, which proves Lemma~\ref{lem-normtracepiexpD2}, and it is used also in \eqref{eq-inequalitygammafunction}, which proves Proposition~\ref{Dixtr}, where the minimum value $d'=1$ for which the integral converges at $t\to 0$ is implicitly used. 

The behavior $I_\nu(t) \,\underset{t\downarrow 0}{\sim}\, \tfrac{C}{\sqrt{t}}$ depends explicitly on the presence of $\fabeta \in \algA$ in \eqref{I(t)}. Indeed, inserting  in \eqref{I(t)} instead an element from the multiplier algebra, like, for instance, $\fabeta(a,\beta) = \bbboneabeta(a,\beta) = \delta_0(a) \notin \algA$, yields  $I_\nu(t) = \tfrac{\sqrt{\pi}}{\sqrt{T}} \int \dd x \,e^{-U x^2} \,\underset{t\downarrow 0}{\sim}\, \tfrac{C'}{t}$. Using this asymptotic behavior in \eqref{tracesym} and \eqref{eq-majorationtracepikernel} would force the minimum value $d'=2$ for the convergence of \eqref{eq-inequalitygammafunction}.

What happens is, that for the integration around $t \sim 0$ in \eqref{eq-inequalitygammafunction}, the behavior of $I_\nu(t)$ is much better with $\fabeta \in \algA$ than with the identity operator or an element of the multiplier algebra. A function $\fabeta \in \algA$ is rapidly decreasing at infinity for the variable $\beta = \nu e^{-x}$, and this behavior at $x \to -\infty$ is reported, through the factor $e^{-T x^2}$ in the integrand, as a better behavior at $t\to 0$ of the integral.  

In \cite{GayrWulk11a}, the same argument can be used to explain the fact that the spectral dimension is half the metric dimension: see Lemmas~7, 10, 12 and Corollary~13 in \cite{GayrWulk11a}.
\end{remark}

%%%%%%%%%%%%%%%%%%%%%%%%%%%%%%%%%%%%%%%%%%%%%%
\section{Conclusions}
%%%%%%%%%%%%%%%%%%%%%%%%%%%%%%%%%%%%%%%%%%%%%%

In a compactified version of $G$, it has been proved in \cite[Theorem 5.2]{IochMassSchu11a} that there are no finitely-summable spectral triples when the representation $\pi$ is quasi-equivalent to the left regular one, while here, Theorem~\ref{spectraldim} breaks this no-go result. An essential difference is that the compactified version of $G$ considered in \cite{IochMassSchu11a} is an antiliminal group (its $C^\ast$-algebra is NGCR) while here the group $G$ is postliminal (Theorem~\ref{thm-postliminal}). Note also that while $\pi_+ \oplus\pi_-$ is quasi-equivalent to the left regular representation, $\pi_0$ is not.  Actually, each of three representations $\pi_\nu$ give rise to a spectral triple of dimension 1 since $\Trw[\pi_\nu(f)(1+\caD^2)^{-1/2}]$ is proportional to $\tau(f)$ for each $\nu \in\{-,0,+\}$. Thus, even if $\pi_0$ does not appear for instance in Plancherel's formula because $\{\pi_-,\,\pi_+\}$ is dense in $\widehat{G}$, this representation $\pi_0=\int_\gR^\oplus dp\,\pi_0^p$ plays the same role as the non-trivial representations $\pi_-$ or $\pi_+$ in the sense that, alone, $\pi_0$ produces a spectral dimension one.

The natural spectral triple on a $2$-dimensional noncommutative torus $\algA_\theta$ is composed of (\cite{Conn96a}, \cite[Section~12.3]{GracVariFigu01a}): the algebra constructed using the space $\caS(\gZ^2)$ of sequences of rapid decay, the Hilbert space $\caH_\tau \vc \ell^2(\gZ^2) \simeq L^2(\gT^2)$ obtained in the GNS construction for the canonical trace $\tau$ on the algebra, and the Dirac operator constructed using the two derivations $\delta_1, \delta_2$ naturally defined on the two unitary generators of the algebra. Using a similar technique as the one used in \ref{subsec-leftregrep}, the GNS representation decomposes as $\pi_\text{GNS} = \int_{\gS^1}^{\oplus} \dd v \, \pi_v$ along irreducible \emph{unitarily nonequivalent} representations $\pi_v$ on the same Hilbert space $L^2(\gS^1, \dd u)$. In this decomposition, $\delta_1$ (resp. $\delta_2$) is represented as a commutator with the derivation along the variable $u$ of the Hilbert space (resp. along the summation variable $v \in \gS^1$).

A similar decomposition can be done for the representation of the noncommutative torus on $L^2(\gR)$ \cite[Section~III.3]{Conn94a}, and it gives similar results: the representation decomposes along irreducible \emph{unitarily nonequivalent} representations on $L^2(\gS^1, \dd u)$, and the first order differential operators $\nabla_1, \nabla_2$ acting on $\caS(\gR) \subset L^2(\gR)$ are mapped to two covariant derivatives, one along the direction of $L^2(\gS^1, \dd u)$, and the other along the summation direction $v$.

In both cases, the derivations $\delta_1, \delta_2$ are represented as commutators with differential operators which act as derivatives in the two directions $u$ and $v$.

The regular representation decomposition given in section \ref{subsec-leftregrep} is quite different to the decompositions obtained for the noncommutative torus, because it takes place along the \emph{same} representation $\pi_-\oplus \pi_+$. As a consequence, any derivative along the composition parameter commutes with the algebra. This explain that the two derivations $\delta_1, \delta_2$ that we consider, can only be represented as commutators with operators along $\caH_-\oplus \caH_+$, as shown in section \ref{subsec-derivationsrepresentations}. The use of the full left regular representation to construct a spectral triple would be useless. Contrary to the noncommutative torus, as far as the action of the Dirac operator on the algebra is concerned, no information can be encoded into the summation parameter. 

We can conclude from Remark~\ref{remark-spectraldim} that the drop in spectral dimension compared to metric dimension is due to the combination of two facts concerning the spectral triple:
\begin{enumerate}[label=\arabic*)]
\item A non-compact noncommutative geometry, which requires to insert $\pi(a)$ in the computations of traces used to evaluate the spectral dimension, with $a \in \algA$ such that its kernel has some properties similar to the ones used in the proof of Lemma~\ref{lem-Ifortsmall}.

\item A harmonic-like Dirac operator $\caD$ (isospectral with the Hamiltonian of a harmonic oscillator): although we have a non-compact noncommutative space, this operator has discrete spectrum and the heat operator is trace-class.

\end{enumerate}
We conjecture that such a drop in spectral dimension can occur in other noncommutative geometries $(\algA, \caH, \caD)$ which share these features, as for instance in \cite{GayrWulk11a}.

There are numerous questions one can ask about the properties of the constructed spectral triple. First of all, it is interesting to ask whether it defines a nontrivial $K$-homology class and in which cyclic cohomology class appears the associated cyclic cocycle (if nontrivial). Finally, one may ask whether further conditions for spectral triples (like Hochschild cycle condition) might be satisfied.

%%%%%%%%%%%%%%%%%%%%%%%%%%%%%%%%%%%%%%%%%%%%%%
\section*{Acknowledgments}
%%%%%%%%%%%%%%%%%%%%%%%%%%%%%%%%%%%%%%%%%%%%%%

We thanks Victor Gayral for his generous help during the preparation of this work. B.~I. and T.~M. would like to thank the \emph{John Templeton Foundation} and the \emph{Copernicus Center for Interdisciplinary Studies} for their financial support during their stay in Krak\'ow.

%%%%%%%%%%%%%%%%%%%%%%%%%%%%%%%%%%%%%%%%%%%%%%
\bibliography{kappaGen}
%%%%%%%%%%%%%%%%%%%%%%%%%%%%%%%%%%%%%%%%%%%%%%

\end{document}